\documentclass[journal,letterpaper,onecolumn,twoside,nofonttune]{IEEEtran}

\usepackage[utf8]{inputenc}
\usepackage[T1]{fontenc}
\usepackage{url}
\usepackage{ifthen}
\usepackage{cite}
\usepackage[cmex10]{amsmath}
\usepackage{balance}

\usepackage{amsmath}
\usepackage{amsthm}
\usepackage{amsfonts}
\usepackage{amssymb}
\usepackage{mathtools}
\usepackage{color}
\usepackage{xcolor}
\usepackage{verbatim}
\usepackage[normalem]{ulem}
\usepackage{enumitem}
\usepackage{hyperref}
\usepackage[ruled,vlined]{algorithm2e}
\usepackage{algpseudocode}

\definecolor{purple}{rgb}{0.5, 0.0, 0.5}
\definecolor{dark_green}{rgb}{0.0, 0.5, 0.0}
\definecolor{dyellow}{rgb}{0.99, 0.93, 0.0}
\definecolor{amaranth}{rgb}{0.9, 0.17, 0.31}

\newcommand{\dgreen}{\color{dark_green}}
\newcommand{\blue}{\color{blue}}
\newcommand{\cyan}{\color{cyan}}

\newcommand{\purple}{\color{purple}}
\newcommand{\white}{\color{white}}

\newcommand{\diffred}{\color{amaranth}}  
\newcommand{\Cc}{\mathcal{C}}
\newcommand{\I}{\mathcal{I}}
\newcommand{\J}{\mathcal{J}}
\newcommand{\K}{\mathcal{K}}
\newcommand{\vb}{\bold{v}}
\newcommand{\Vb}{\bold{V}}
\newcommand{\wb}{\bold{w}}

\newcommand{\xb}{\bold{x}}
\newcommand{\Xb}{\bold{X}}
\newcommand{\Xbh}{\hat{\bold{X}}}

\newcommand{\Yb}{\bold{Y}}
\newcommand{\Gb}{\bold{G}}
\newcommand{\Tb}{\bold{T}}
\newcommand{\texS}{\mathrm{S}}
\newcommand{\texT}{\mathrm{T}}
\newcommand{\D}{\mathcal{D}}
\newcommand{\T}{\mathcal{T}}
\newcommand{\At}{\tilde{A}}
\newcommand{\Bt}{\tilde{B}}
\newcommand{\Bbt}{\tilde{\bold{B}}}
\newcommand{\Ct}{\tilde{C}}

\newcommand{\Abar}{\bar{A}}
\newcommand{\Bbar}{\bar{B}}
\newcommand{\Cbar}{\bar{C}}

\newcommand{\Z}{\mathbb{Z}}
\newcommand{\R}{\mathbb{R}}
\newcommand{\C}{\mathbb{C}}
\newcommand{\F}{\mathbb{F}}
\newcommand{\E}{\mathbb{E}}
\newcommand{\N}{\mathbb{N}}
\newcommand{\Q}{\mathbb{Q}}
\newcommand{\gb}{\bold{g}}
\newcommand{\gt}{\tilde{g}}

\newcommand{\Ab}{\bold{A}}
\newcommand{\Bb}{\bold{B}}
\newcommand{\Cb}{\bold{C}}
\newcommand{\Cbbar}{\bar{\Cb}}
\newcommand{\Ib}{\bold{I}}
\newcommand{\ab}{\bold{a}}
\newcommand{\abt}{\tilde{\bold{a}}}
\newcommand{\bb}{\bold{b}}
\newcommand{\eb}{\bold{e}}
\newcommand{\rt}{\tilde{r}}
\newcommand{\Cfr}{\mathfrak{C}}
\newcommand{\taut}{\tilde{\tau}}
\DeclareMathOperator{\spn}{span}
\DeclareMathOperator{\supp}{supp}
\DeclareMathOperator{\nnz}{nnzr}
\DeclareMathOperator{\rank}{rank}
\DeclareMathOperator{\rem}{rem}
\DeclareMathOperator{\tr}{tr}
\DeclareMathOperator{\cmmO}{\mathrm{CMM-1}}
\DeclareMathOperator{\cmmT}{\mathrm{CMM-2}}
\DeclareMathOperator{\OPR}{\mathrm{(OP-R)}}

\DeclareMathOperator{\IPB}{\mathrm{(IP-B)}}
\newcommand{\ldpc}{\sfsty{LDPC}}
\newcommand{\SVD}{\mathrm{SVD}}

\DeclarePairedDelimiter\ceil{\lceil}{\rceil}
\DeclarePairedDelimiter\floor{\lfloor}{\rfloor}

\DeclarePairedDelimiter\bceil{\Big\lceil}{\Big\rceil}
\DeclarePairedDelimiter\bfloor{\Big\lfloor}{\Big\rfloor}

\newtheorem{Thm}{Theorem}
\newtheorem{Cor}[Thm]{Corollary}
\newtheorem{Prop}[Thm]{Proposition}
\newtheorem{Lemma}[Thm]{Lemma}
\newtheorem{Def}[Thm]{Definition}
\newtheorem{Rmk}[Thm]{Remark}

\newcommand{\ind}{\text{\color{white}.$\quad$}}
\DeclareMathOperator*{\argmin}{arg\,min}
\newcommand{\sfsty}[1]{\ensuremath{\mathsf{#1}}}  

\begin{document}
\title{Generalized Fractional Repetition Codes for Binary Coded Computations}

\author{\large%
Neophytos Charalambides, Hessam Mahdavifar, and Alfred O. Hero III
  \vspace{-.25in}
  \thanks{Part of the material in this paper was presented at the 2020 IEEE International Symposium on Information Theory (ISIT), Los Angeles, USA, June 2020. \newline
  ${\white..}$ Neophytos Charalambides is with the Department of Computer Science and Engineering, University of California at San Diego, La Jolla, CA 92093 USA.\newline
  ${\white..}$ Hessam Mahdavifar is with the Department of Electrical and Computer Engineering, Northeastern University, Boston, MA 02115 USA.\newline
  ${\white..}$ Alfred O. Hero, III is with the Department of Electrical Engineering and Computer Science, University of Michigan at Ann Arbor, Ann Arbor, MI 48104 USA.\newline
  ${\white..}$ e-mails: \texttt{neophytoschara@ucsd.edu}, \texttt{h.mahdavifar@northeastern.edu}, \texttt{hero@umich.edu}
  }
}

\maketitle

\begin{abstract}
This paper addresses the gradient coding and coded matrix multiplication problems in distributed optimization and coded computing. We present a computationally efficient coding method which overcomes the drawbacks of the \textit{Fractional Repetition Coding} gradient coding method proposed by Tandon et al., and can also be leveraged by coded computing networks whose servers are of heterogeneous nature. Specifically, we propose a construction for fractional repetition gradient coding; while ensuring that the generator matrix remains close to perfectly balanced for any set of coding parameters, as well as a low complexity decoding step. The proposed binary encoding avoids operations over the real and complex numbers which inherently introduce numerical and rounding errors, thereby enabling accurate distributed encodings of the partial gradients. We then make connections between gradient coding and coded matrix multiplication. Specifically, we show that any gradient coding scheme can be extended to coded matrix multiplication. Furthermore, we show how the proposed binary gradient coding scheme can be used to construct two different coded matrix multiplication schemes, each achieving different trade-offs.
\end{abstract}
\vspace{3mm}


\section{Introduction}
\label{intro}

The \textit{curse of dimensionality} has been a major impediment to solving large scale problems, which often require heavy computations. Recently, coding-theoretic ideas have been adopted in order to accommodate such computational tasks in a distributed fashion, under the assumption that \textit{straggler} workers are present \cite{LLPPR18,li2016unified,reisizadeh2017coded,li2016coded,li2017coding,yang2017computing,vulimiri2013low,wang2018coded,ramamoorthy2019universally,YSRKSA18,RRG20,SMA21,MCSPJ22}. Stragglers are workers whose assigned tasks may never be completed, due to delay or outage, and can significantly increase the overall computation time. When the computational tasks are encoded, e.g., by using linear codes, the distributed computations can be protected against erasures.

In this paper, we first focus on the problem of exact recovery of the gradient in a distributed computing setting. We adopt the framework of \cite{TLDK17} where \textit{Gradient Coding} (GC)\footnote{For brevity, we refer to a \textit{gradient coding scheme} as GCS, a \textit{coded matrix multiplication (scheme)} as CMM/CMMS; and CMMSs for plural.} was proposed for recovery of the gradient when the objective loss function is differentiable and additively separable. Gradient recovery is accomplished by replicating the tasks in a certain way, so as to introduce the redundancy needed to combat the effect of stragglers. The problem of exact recovery of the gradient was studied in several prior works, e.g., \cite{TLDK17,HASH17,RTTD17,OGU19,CMH20,LMSAS18}, while the numerical stability issue was studied in \cite{YA18}. There are also several works involving GC for approximate recovery of the gradient \cite{CP18,RTTD17,CPE17,WCP19,BWE19,WLS19,KKR19,HYKM19,CHZP18,CPH20a,CMPH22}. Also, an idea similar to GC had appeared in \cite{ZKJKS08}, though not within a coding-theoretic setting. 

We propose a scheme for GC that is numerically accurate and efficient, and works in fixed point precision. The proposed scheme avoids floating-point representations and operations, e.g., division or multiplication of real or complex numbers. Furthermore, the encoding matrix is binary, simplifying the encoding process. This scheme is also deterministic and does not require generating random numbers. The method is similar in spirit to the \textit{fractional repetition scheme} introduced in \cite{ERK10,TLDK17}, where we also drop the strict assumption that $s+1$ divides $n$, where $n$ is the number of workers and $s$ is the number of stragglers that the scheme tolerates. The main advantage of encoding and decoding real-valued data using binary matrices is that it does not introduce further numerical errors, possibly adding to the rounding errors of the associated computation tasks. Such a binary approach was considered in \cite{JSM19} for matrix-vector multiplication. The fact that the encoding matrix is defined over $\{0,1\}$ allows us to view the encoding as task assignments. This also leads to a more efficient online decoding, which avoids searching through a polynomially large table in the number of workers $n$, as in the original GCS proposed in \cite{TLDK17}. Moreover, our encoding matrix can be understood as a generalization of the \textit{Fractional Repetition Coding} (FRC) scheme from \cite{TLDK17} for the case when $(s+1)\nmid n$, and our decoding algorithm can be used in conjunction with the corresponding GCS, for a low complexity online decoding without constructing a large decoding matrix a priori.

Dropping the aforementioned assumption results in an unbalanced load assignment among the workers. Under the assumption that the workers are \textit{homogeneous}, we allocate the partitions as uniform as possible. To avoid bias when considering workers of \textit{heterogeneous} nature, i.e., of different computational power, the allocation of the computational tasks should be done in such a way that all workers have the same expected completion time; as the stragglers are assumed to be uniformly random. We provide an analysis which determines how to appropriately allocate the assignments, so that this objective is met. We note that similar ideas appear in \cite{TLDK17, OGU19}, in the context of \textit{partial} and \textit{non-persistent} stragglers; respectively.

The majority of coded computing has focused on fundamental algebraic operations, such as matrix multiplication and polynomial evaluation. Directly adopting these schemes in general optimization problems is often not possible, since the gradient computation may not have any algebraic structure, or can only be evaluated numerically \cite{LA20}. In this work we study the other direction, i.e., how to devise \textit{Coded Matrix Multiplication} (CMM) schemes from GC. The key idea is to leverage the additive structure underlying both problems. By a simple modification to the encoding and decoding steps of any exact GCS, we show that the GCS can be utilized for matrix multiplication. In a similar fashion, we can transform any GCS into a distributive straggler robust addition scheme.

The proposed GC method can be adapted to compute matrix-matrix multiplication in the presence of stragglers; which has gained a lot of attention recently, as well as matrix inverse approximations \cite{CPH20b,CPH22a}. The first CMMS was proposed in \cite{LLPPR18}. Since then, a multitude of CMMSs have been proposed \cite{YMAA17,LSR17,YMAA20,YA20,FJHDCG17,DFHJCG19,FC19,SHN19,DRV21,RT21,DR21}, with each of them being advantageous to others in certain aspects. There is also a considerable amount of work on coded matrix-vector multiplication in distributed computing \cite{DCG16,FC18,JSM19,BP23}. Furthermore, numerical stability for coded matrix multiplication has been studied in \cite{FC19,SHN19,DRV21,RT21,DR21,JDCC21}. Approximate coded matrix-matrix and matrix-vector schemes have also been devised \cite{FD16,JM19,CPH20c,THRD21,KD22,JDCC21,RCHV23}. It is also worth mentioning tangentially related works to coded computing \cite{KDKE23,BWCE23,MDSE24,KE24}, which do not consider the presence of stragglers.

We show that any GCS can be extended to a CMMS. The main idea is that the product of two matrices is equal to the sum of the outer-products of their columns and rows, respectively. This property has been utilized in the context of CMM \cite{FJHDCG17,DFHJCG19,CPH20c} and, to our knowledge, we are the first to connect this to GC. We present two new CMMSs based on the proposed binary GC, each achieving different trade-offs. Since the proposed CMMSs are derived from a GCS, they have properties that differ from other CMM approaches, and do not satisfy the same bounds and thresholds. However, our proposed schemes achieve the optimal trade-off between task allocations and the number of stragglers of GC. They also preserve the desired properties possessed by binary GC methods. For example, there is no need for complex encoding and decoding procedures; and the proposed methods are numerically accurate and computationally efficient.

Our main contributions are the following:
\begin{itemize}[noitemsep,nolistsep]
  \item Introduction of a \textit{binary} GCS --- both in the encoding and decoding, that is resilient to stragglers;
  \item Elimination of the restrictive assumption $(s+1)\mid n$ in \cite{TLDK17};
  \item We show that the proposed GCS achieves perfect gradient recovery;
  \item We derive the minimum maximum load over all workers of any binary fractional repetition GCS, for any pair of parameters $(s,n)$;
  \item We show how the unbalanced assignment, which arises when $(s+1)\nmid n$, can be made optimal when the workers are homogeneous;
  \item As compared to the original binary scheme \cite{TLDK17}, we give a more efficient online decoding step;
  \item We determine the optimal task assignment for heterogeneous workers;
  \item We show how any GCS can be extended to a CMMS;
  \item We use our binary GCS to devise two binary CMMSs.
\end{itemize}

The rest of this paper is organized as follows. In Section \ref{str_problem_GC} we provide a review of the straggler problem in GC \cite{TLDK17}. In Section \ref{BGC_sec} binary GC is introduced, and we describe the conditions for which a close to balanced task allocation needs to meet; when $(s+1)\nmid n$. Our binary encoding and decoding procedures are discussed in Section \ref{proposed_scheme_section}. In Subsection \ref{val_opt_subsec} we establish the optimality of our GCS, and in Subsection \ref{heter_case_sec} we consider scenarios with heterogeneous workers. The focus is then shifted towards CMM, and in Section \ref{matr_mult_sec} we show how any GCS can be utilized to devise a CMMS. Then, another CMMS derived from our GCS is presented in Subsection \ref{2nd_sch_sec}. In Subsection \ref{comp_section} we compare and contrast our two CMMSs, and discuss where they have been utilized in other coded computing applications. In Section \ref{comparison_sec} we compare our schemes to prior methods, and draw connections with other areas in information theory. Section \ref{sec:conc} concludes the paper.

We also provide appendices with further details on our algorithms, numerical examples, and experimental justification. In Appendix \ref{appl_Fr_norm} we present various applications in which CMM can be utilized in gradient descent iterative algorithms; for Frobenius-norm minimization problems. These demonstrate further connections between the CMM and GC problems.

\section{Preliminaries}
\label{str_problem_GC}

\subsection{Straggler Problem}
\label{str_problem}

Consider a single central server that has at its disposal a dataset $\D=\left\{(\xb_i,y_i)\right\}_{i=1}^N\subsetneq \R^p\times\R$ of $N$ samples, where $\xb_i$ represents the features and $y_i$ denotes the label of the $i^{th}$ sample. The central server distributes the dataset $\D$ among $n$ workers to facilitate computing the solution of the problem
\begin{equation}
\label{th_star_pr}
  \theta^{\star} = \arg\min_{\theta\in\R^p}\left\{ \sum_{i=1}^N \ell(\xb_i,y_i;\theta) + \mu R(\theta) \right\}
\end{equation}
where $L(\D;\theta)=\sum_{i=1}^N \ell(\xb_i,y_i;\theta)$ is the empirical loss; for $\ell(\xb_i,y_i;\theta)$ a predetermined differentiable loss-function, and $\mu R(\theta)$ is a regularizer. A common approach to solving (\ref{th_star_pr}) is to employ gradient descent. Even if closed-form solutions exist for (\ref{th_star_pr}), gradient descent is advantageous for large $N$.

The central server is assumed to be capable of distributing the dataset appropriately, with a certain level of redundancy, in order to recover the gradient based on the full dataset $\D$. For clarity, we will assume that $k\mid N$. As a first step we partition $\D$ into $k$ disjoint parts $\{\D_j\}_{j=1}^k$ each of size $N/k$. If $k\nmid N$ such that $N=\alpha k+\beta$ with $\beta\equiv N\bmod k \not\equiv 0$, we can partition $\D$ into $\beta$ parts of size $\ceil{N/k}=\alpha+1$ and $k-\beta$ parts of size $\floor{N/k}=\alpha$. The gradient is the quantity
\begin{equation}
\label{grad_equation}
  g=\nabla_{\theta}L(\D;\theta)=\sum_{j=1}^k\nabla_{\theta}\ell(\D_j;\theta)=\sum_{j=1}^k g_j\ .
\end{equation}
We refer to the terms $g_j\coloneqq\nabla_{\theta}\ell(\D_j;\theta)$ as \textit{partial gradients}. 

In a distributed computing setting each worker node completes its task by returning a certain encoding of its assigned partial gradients. There can be different types of failures that may occur during the computation or the communication process. The worker nodes that fail to complete their tasks and return the outcome to the central server are called \textit{stragglers}. It is assumed that there are $s$ stragglers, thus, the central server only receives $f=n-s$ completed tasks. Let $\I\subsetneq\N_n\coloneqq\{1,\cdots,n\}$ denote the set of indices of the $f$ workers who complete and return their tasks. In practice, the completed tasks may be received at different times. Once \textit{any} set of $f$ tasks is received, the central server should be able to decode the received encoded partial gradients and recover the full gradient $g$.

\subsection{Gradient Coding}
\label{GC}

Gradient coding, proposed in \cite{TLDK17}, is a procedure comprised of an encoding matrix $\Bb\in\Sigma^{n\times k}$, and a decoding vector $\ab_{\I}\in\Sigma^n$; determined by $\I$, for $\Sigma$ the field over which the encoding-decoding takes place. It is commonly assumed that the workers have the same computational power, in which case the same number of tasks is assigned to each of them. We relax this restriction in this paper, and thus do not need to impose the assumption $(s+1)\mid n$ from \cite{TLDK17}. Each row of $\Bb$ corresponds to an encoding vector, also regarded as a task allocation, and each column corresponds to a data partition $\D_j$.

Each worker node is assigned a number of partial gradients from the partition $\{\D_j\}_{j=1}^k$, indexed by $\J_i\subsetneq\N_k$. The workers are tasked to compute an encoded version of the partial gradients $g_j\in\R^p$ corresponding to their assignments. Let
\begin{equation}
\label{g_matrix}
  \gb \coloneqq {\begin{pmatrix} | & | & & | \\ g_1 & g_2 & \hdots & g_k \\ | & | & & | \end{pmatrix}}^T \in \R^{k\times p}
\end{equation}
denote the matrix whose rows constitute the transposes of the partial gradients. The received encoded partial gradients will be the rows of $\Bb\gb\in\R^{n\times p}$ indexed by $\I$.

The full gradient of the objective (\ref{th_star_pr}) on $\D$ can be recovered by applying $\ab_{\I}$ which is designed to have a support that is a subset of $\I$
\begin{equation}
\label{GC_identity}
  g^T=\ab_{\I}^T(\Bb\gb)=\bold{1}_{1\times k}\gb= \sum_{j=1}^kg_j^T,
\end{equation}
provided that the encoding matrix $\Bb$ satisfies
\begin{equation}
\label{GC_condition}
  \ab_{\I}^T\Bb=\bold{1}_{1\times k}
\end{equation}
for all ${{n}\choose{s}}$ possible index sets $\I$. Note that in perfectly balanced schemes, every partition is sent to $s+1$ servers, and each server will receive at least $\frac{k}{n}(s+1)$ distinct partitions. In Section \ref{BGC_sec}, we propose a binary design of the encoding matrix $\Bb$ and decoding vector $\ab_{\I}$. These may then be used for recovering the gradient $g$ at each iteration by the central server.

In \cite{TLDK17}, a \textit{balanced} assignment is considered, which is the case where all the workers are assigned the same number of tasks. This number is lower bounded by $\frac{k}{n}(s+1)$, i.e.,
\begin{equation}
\label{bound_TLDK}
  \|\Bb_{(i)}\|_0 \geq \frac{k}{n}(s+1) \quad \text{ for all } i\in\N_n
\end{equation}
for $\Bb_{(i)}$ the $i^{th}$ row of $\Bb$. When this lower bound is met with equality, the scheme is\textit{ maximum distance separable} (MDS). The restriction $(s+1)\mid n$ allows the GC to satisfy this bound, as $\frac{n}{s+1}$ needs to be an integer. Since the bound of \eqref{bound_TLDK} implies that any balanced scheme with encoding-decoding pair $(\Bb,\ab_\I)$ that assigns the same number of data samples to all the workers must assign at least $\frac{s+1}{n}$ fraction of the data to each worker, which is independent of $k$, the analysis throughout \cite{TLDK17} assumes that $n=k$. We will make the same assumption for our main analysis, in order to also give explicit bounds which do not involve the floor or ceiling functions. However, in Subsection \ref{close_bal_des_subsec}, we give a description of how one handles the general cases where $n\neq k$, and how we can determine an appropriate partitioning of $\D$ into $n$ parts when considering homogeneous workers. In Theorem \ref{thm_mim_max_load}, we give an analogous lower bound to \eqref{bound_TLDK} for when $(s+1)\nmid n$.

In general, GC schemes require processing over $s+1$ partitions from each worker in order to tolerate up to $s$ stragglers. Since these schemes encode over partial gradients computed from unprocessed data, they are applicable to a large class of loss functions, as well as loss functions whose gradient can only be computed numerically, e.g., in deep neural networks \cite{LA20}.

The binary GCS proposed in \cite{TLDK17}, replicates the task done by a subset of the workers, following the steps below:
\begin{enumerate}
  \item divide the $n$ workers into $(s+1)$ groups of size $\frac{n}{s+1}$;
  \item  in each group, all the data is divided equally and disjointly, assigning $(s+1)$ partitions to each worker;
  \item all the groups are replicas of each other;
  \item when finished computing, every worker transmits the sum of its partial gradients.
\end{enumerate}
We note that 1) and 2) correspond to what is referred to as ``fractional'', and 3) to ``repetition'', in the choice of the GCS's name. Point 4), corresponds to the local encoding performed by the workers. Such coding schemes were first considered in \cite{ERK10} to achieve storage capacity for random access repair, in distributed storage systems. The encoding matrix in \cite{TLDK17} is constructed by first defining a binary encoding matrix for each of the groups, i.e., point 2):
\begin{equation*}
\label{B_block}
  \Bb_{\text{block}} = \Ib_{\frac{n}{s+1}}\otimes\bold{1}_{1\times(s+1)}\in\{0,1\}^{\frac{n}{s+1}\times n}
\end{equation*}
and then augmenting the binary encoding matrices $\{\Bb_{\text{block}}^{(i)}\}_{i=1}^{s+1}$ for each worker $i$; all of which are equal to $\Bb_{\text{block}}$, which corresponds to point 3):
\begin{equation}
\label{B_FRC}
  \Bb_{\text{FRC}} = \bold{1}_{(s+1)\times1}\otimes\Bb_{\text{block}} = \begin{bmatrix} \Bb_{\text{block}}^{(1)} \\ \vdots \\ \Bb_{\text{block}}^{(s+1)} \end{bmatrix} \in \{0,1\}^{n\times n}\ .
\end{equation}

In what we present, we consider an equivalent encoding\footnote{By \textit{equivalent encoding}, we refer to two encodings in which one can be obtained from the other via a permutation of the generator's rows.} for FRC to simplify our analysis and decoding, which is done in terms of congruence classes. This encoding is a permutation on the rows of $\Bb_{\text{FRC}}$, and is defined as follows:
\begin{equation*}
\label{B_FRC_permuted}
  \Bb_{\text{FRC}}' = \Ib_{\frac{n}{s+1}} \otimes \bold{1}_{(s+1)\times(s+1)}\ .
\end{equation*}
Simply stated, $\Bb_{\text{FRC}}'$ is a block diagonal matrix with blocks of size $(s+1)\times(s+1)$, comprised of all ones. This further elaborates as to why the proposed coding technique is called `Fractional Repetition Coding'. Specifically, a fraction of the data (support across the rows of $\Bb_{\text{FRC}}'$) are sent to each of the workers, and the same repeated encoding of partial gradients is requested from a subset of the servers.

We can now formally define what we mean by fractional repetition coding; and generalized fractional repetition coding, in coded computing.

\begin{Def}
\label{FRC_def}
In the context of coded computing; \textbf{fractional repetition coding} is any coded computing scheme which admits an encoding matrix of the same or an equivalent structure as \eqref{B_FRC}.
\end{Def}

As we have already seen, for an encoding scheme to be able to have the structure of \eqref{B_FRC}; it is necessary that $(s+1)\mid n$. In the next definition, we are able to relax the definition of fractional repetition coding by removing this stringent assumption, and hence generalize Definition \ref{FRC_def}. We clarify that generalized FRC can be defined for any pair of integers $(s,n)$ where $0\leq s<n$, and in the case where $(s+1)\mid n$, we get a coding scheme according to Definition \ref{FRC_def}. An explicit example where $(s+1)\nmid n$ is provided in Appendix \ref{example_app}.

\begin{Def}
\label{gen_FRC_def}
In the context of coded computing, \textbf{generalized fractional repetition coding} is any coded computing scheme which has an encoding matrix that is the augmentation of two FRC encoding matrices $\Bb_1$ and $\Bb_2$ with a structure similar to that of \eqref{B_FRC}; where the only difference is that now no two rows differ in cardinality of their support by more than one within each $\Bb_1$ and $\Bb_2$, and no two rows across both $\Bb_1$ and $\Bb_2$ in cardinality of their support by more than $t+2$; for $t=\Big\lfloor n\bmod(s+1)\big/\big\lfloor\frac{n}{s+1}\big\rfloor \Big\rfloor$.
\end{Def}

By a fractional repetition GCS, we mean a GCS with a FRC encoding matrix. To avoid repetitiveness, henceforth by fractional repetition GCS we mean an encoded GC method that uses a \textit{generalized} FRC encoding matrix. Analogously, we define a fractional repetition CMMS. We refer to any coded computing method with an encoding matrix equivalent to the ones defined in Definition \ref{gen_FRC_def}, as a fractional repetition method.

At this point, we also note that in related work \cite{LKYSA18} the partitions sent to each worker are pre-processed, such that the computations at the workers are viewed as evaluating a polynomial at distinct points. This approach is referred to as \textit{Polynomially coded regression}, and only applies to the least squares objective function. The central server computes the gradient by interpolating this polynomial. By working on the encoded data instead, the authors of \cite{LKYSA18} reduce the threshold on the number of workers that need to respond.

\subsection{Notational Conventions}
\label{notation_sec}

Let $\vb\in\R^p$ and $\Vb\in\R^{p\times q}$ respectively denote an arbitrary vector and matrix of the specified dimensions. The support of $\vb$; i.e., the index set of elements which are nonzero, is denoted by $\supp(\vb)$. The number of nonzero elements in $\vb$ is denoted by $\nnz(\vb)$; i.e., $\nnz(\vb)=|\supp(\vb)|$. We define $\nnz(\Vb)$ analogously. The row-span of a $\Vb$ is denoted by $\spn(\Vb)$.

The vector Euclidean norm of $\vb$ is defined as $\|\vb\|_2 = \sqrt{\vb^T\vb} = 
\left(\sum_i\vb_i^2\right)^{1/2}$, the $L_0$ norm of $\vb$ as $\|\vb\|_0=\nnz(\vb)$, and the matrix Frobenius norm of $\Vb$ as $\|\Vb\|_F = \sqrt{\tr(\Vb^T\Vb)} = \left(\sum_i\sum_j\Vb_{ij}^2\right)^{1/2}$. Also, $\Vb_{(i)}$ denotes the $i^{th}$ row of $\Vb$, $\Vb^{(j)}$ denotes the $j^{th}$ column of $\Vb$, and the $p \times p$ identity matrix is denoted by $\Ib_p$. For a row index set $I\subseteq\{1,2,\ldots,p\}$ of $\Vb$, the submatrix comprised of the rows indexed by $I$ is $\Vb_{I}\in\R^{|I|\times q}$. We denote the set of nonnegative integers by $\N_0\coloneqq\{0,1,2,\ldots\}$, the positive integers up to $n$ by $\N_n\coloneqq\{1,\ldots,n\}$; the positive integers up to $n$ and including $0$ by $\N_{0,n}\coloneqq\{0,1,\ldots,n\}$, and the collection of size $q$ subsets of $\N_n$ by $\I_q^n$; i.e., $|\I_q^n|={{n}\choose{q}}$. Disjoint unions are represented by $\bigsqcup$; e.g., $\Z=\{j:j\text{ is odd}\}\bigsqcup\{j:j\text{ is even}\}$. By $\eb_j$ we denote the $j^{th}$ standard basis vector of $\R^n$. The remainder function is denoted by $\rem(\cdot,\cdot)$, i.e., for positive integers $a$ and $b$; $\rem(b,a)=b-a\cdot\floor{\frac{b}{a}}$.

In the context of GC, the parameter $N$ is reserved for the number of data samples, each consisting of $p$ features and one label. In the context of CMM, the parameter $N$ denotes the common dimension of the two matrices being multiplied, i.e., the number of columns of the first matrix and the number of rows of the second matrix. The integer $k$ denotes the number of partitions of the dataset in GC, and of the matrix or matrices in CMM. With $n$ the number of workers and $s$ the number of stragglers, the number of ``blocks'' is determined by $\ell=\floor{\frac{n}{s+1}}$. The set of indices of the $f=n-s$ non-straggling workers is denoted by $\I$, and is an element of $\I_f^n$. Our encoding matrix is denoted by $\Bb$, and our decoding vector for the case where a certain $\I$ occurs is denoted by $\ab_\I$. In the CMM setting, these are denoted by $\Bbt$ and $\abt_\I$, respectively.

\section{Binary Gradient Coding}
\label{BGC_sec}
\setcounter{MaxMatrixCols}{20}

In this section we motivate our approach to binary GC. The main idea behind binary schemes is to ensure that the superposition of the corresponding encoding vectors of a certain subset of the workers with non-overlapping assigned tasks, results in the all ones vector. This gives us a simple condition for binary GC, which leads to a special case of condition \eqref{GC_condition}. Additionally, we derive a strict lower bound for the total computational load of any GCS, and formalize what we mean by ``close to uniform/balanced'' assignments. It is worth noting that this problem has also been extensively studied, e.g., \cite{ABKU99}. We incorporate these into an optimization problem with rational constraints, which we constructively solve through our encoding in Section \ref{proposed_scheme_section}; when requiring the additional constraint that the encoding-decoding pair is over $\{0,1\}$. We also derive the minimum maximum load over all workers of any binary fractional repetition GCS, for any pair of parameters $(s,n)$.

\subsection{Binary GC Condition}

For our encoding, we have the following simple strategy in order to meet condition \eqref{GC_condition}, for any $\I\in\I_f^n$. We divide the workers into $s+1$ subsets, and arrange the data partitions among the workers of each subset in such a way, so that their allocated task of computing and encoding certain partial gradients (corresponding to the arrangement); in each subset partition the entire gradient without any overlaps within the partition of the workers. If there were any overlaps, the corresponding partial gradients would be accounted for more than once, as the decoding is binary. The partitions of the worker subsets are indexed by $s+1$ disjoint sets $\{\K_i\}_{i=0}^{s}$; i.e., $\bigsqcup_{i=0}^{s}\K_i=\N_n$. This is a structural assumption on the encoding matrix of FRC schemes, where each $\K_i$ corresponds to $\Bb_{\text{block}}^{(i+1)}$ in \eqref{B_FRC}. In our approach, $\{\K_i\}_{i=0}^{s}$ correspond to distinct congruence classes, which we describe in Section \ref{proposed_scheme_section}.

As soon as all workers from one of the $s+1$ worker subsets have responded, the gradient $g$ is recoverable. In the worst case, we will have $n-s$ responses, as by the pigeonhole principle; at most $s$ subsets will have exactly one straggler, and the remaining subset will have none. From this, in the case where the workers are partitioned into more than $s+1$ subsets, the scheme could tolerate more stragglers. By this, since we are considering a fixed $s$, we divide the workers into exactly $s+1$ subsets. Furthermore, the arrangement of data partitions which takes place within each subset of workers, does not matter, as long as every data block/partial gradient is assigned to exactly one worker, i.e., there is no overlap. In what follows, we do not consider the degenerate case where a worker is not assigned any partition, i.e., $|\supp(\Bb_{(i)})|\geq1$ for all $i\in\N_n$. This idea is summarized in Proposition \ref{general_cond_prop} and Corollary \ref{eq_formulation_cor}.

\begin{Prop}
\label{general_cond_prop}
Let $\Bb\in\{0,1\}^{n\times k}$, and partition its rows into $s+1$ nonempty subsets with index sets $\{\K_i\}_{i=0}^{s}$; i.e., $\bigsqcup_{i=0}^{s}\K_i=\N_n$. If for all $i\in\N_{0,s}$:
\begin{equation}
\label{cond_gen_GC_B}
  \sum_{j\in\K_i}\Bb_{(j)}=\bold{1}_{1\times k},
\end{equation}
then, for any $\I\in\I_f^n$, it follows that $\bold{1}_{1\times k}\in\spn(\Bb_\I)$. This is a sufficient condition for binary encoding-decoding pairs $(\Bb,\ab_\I)\subsetneq\{0,1\}^{n\times k}\times\{0,1\}^n$, to satisfy \eqref{GC_identity}.
\end{Prop}

\begin{proof}
Fix an arbitrary $\I\subsetneq\N_n$ of size $f=n-s$. By the pigeonhole principle, since only $s$ indices from $\N_n$ are not included in $\I$, we know that at least one of $\{\K_i\}_{i=0}^{s}$ is contained in $\I$; say $\K_l$. This implies that $\Bb_{\K_l}\in\{0,1\}^{|\K_l|\times k}$ is a submatrix of $\Bb_\I\in\{0,1\}^{f\times k}$. By condition \eqref{cond_gen_GC_B}, it follows that $\bold{1}_{1\times k}\in\spn(\Bb_{\K_l})\subsetneq\spn(\Bb_\I)$. Hence, this is a sufficient condition for satisfying \eqref{GC_identity} with a binary encoding-decoding pair, which completes the proof.
\end{proof}

\begin{Cor}
\label{eq_formulation_cor}
An equivalent formulation of \eqref{cond_gen_GC_B}, is to simultaneously satisfy: $\supp(\Bb_{(j)})\bigcap\supp(\Bb_{(l)})=\emptyset$ for all $j,l\in\K_i$ where $j\neq l$, and $\bigcup_{\iota\in\K_i}\supp(\Bb_{(\iota)})=\N_k$. Moreover, the corresponding decoding vector used to meet \eqref{GC_condition} when $\K_i\subsetneq \I$, is $\ab_{\I}=\sum_{j\in\K_i}\eb_{j}\in\{0,1\}^n$.
\end{Cor}

\begin{proof}
Assume we have a binary GCS with encoding matrix $\Bb$ and the index set of responsive workers $\I$, for which $\K_i\subsetneq\I$. For a contradiction, assume that for $\K_i$ we have $h\in\supp(\Bb_{(j)})\bigcap\supp(\Bb_{(l)})$. It then follows that in the $h^{th}$ entry of $\sum_{j\in\K_i}\Bb_{(j)}$, we have an integer greater than 1, which violates \eqref{cond_gen_GC_B}. Furthermore, under the assumption that $\bigcup_{\iota\in\K_i}\supp(\Bb_{(\iota)})=\N_k$, it is straightforward that \eqref{cond_gen_GC_B} holds.

We now need to show that \eqref{cond_gen_GC_B} implies the two conditions. Under the assumption that \eqref{cond_gen_GC_B} is true, it follows for each $h\in\N_k$; there is one and only one $j\in\K_i$ for which the $h^{th}$ entry of $\Bb_{(j)}$ is equal to one. If there were $H>1$ many such $j$'s or none, then the $h^{th}$ entry of $\sum_{j\in\K_i}\Bb_{(j)}$ would be $H$ or $0$ respectively, contradicting \eqref{cond_gen_GC_B}.

Since we also require the decoding vector $\ab_\I$ to be binary, we can either add (without rescaling) or ignore the computations of the workers within the given subgroup. Since the objective is to meet \eqref{GC_condition}, by \eqref{cond_gen_GC_B} we simply need to sample and add the corresponding encoding rows of $\K_i$. This is done by the decoding vector $\ab_{\I}=\sum_{j\in\K_i}\eb_{j}$.
\end{proof}

The following lemma is a direct generalization of \eqref{bound_TLDK} from \cite[Theorem 1]{TLDK17}, which considers the balanced case. Our proposed scheme meets the lower bound with equality, which implies a minimized total computational load across the network, i.e., $\Bb$ is as sparse as possible for a GCS that is resilient to $s$ stragglers. We attain a minimal total load balance, while ensuring that we are as balanced as possible when $(s+1)\nmid n$.

\begin{Lemma}
\label{min_load_lem}
The total computational load of any GCS that is resilient to $s$ stragglers; is at least $k\cdot(s+1)$, i.e., $\nnz(\Bb)\geq k\cdot(s+1)$.
\end{Lemma}

\begin{proof}
In order to tolerate $s$ stragglers, each of $\{\D_j\}_{j=1}^k$ needs to be allocated to at least $s+1$ workers, thus $\|\Bb^{(i)}\|_0\geq s+1$ for each $i\in\N_k$. Since we have $k$ partitions, it follows that the total load is at least $k\cdot(s+1)$.
\end{proof}

\subsection{Close to Uniform Assignment Distribution}
\label{cl_un_ass_distr_sec}

A drawback of the GCS proposed in Section \ref{BGC_sec} is that the load assignments can have a wide range depending on how small $r$ is compared to $s+1$. This is due to the lighter load assigned to the workers in the remainder block; which is of size $r$. The uneven workload is the cost we pay for dropping the assumption $(s+1)\mid n$, which does not often hold for a pair of two arbitrary positive integers $(s,n)$\footnote{For fixed $n$ and random $s\in\{0,\cdots,n-1\}$; we have $(s+1)\mid n$ with probability $\frac{\sigma_0(n)-2}{n}$, where $\sigma_0$ is the divisor function of the $0^{th}$ power.}. It is worth mentioning that for small $s$ and a fixed $n$ for which $(s+1)\nmid n$, when considering the original FRC scheme \cite{TLDK17}, one can easily modify $s'$ and decrease $n$ to $n'$ in order to meet the divisibility condition $(s'+1)\mid n'$.

In order to have a close to balanced GCS, we wish that the partitioning from Proposition \ref{general_cond_prop} is done so that $\big||\K_j|-|\K_l|\big|\leq1$; for all $j,l\in\N_{0,s}$, and $\big|\supp(\Bb_{(j')})-\supp(\Bb_{(l')})\big|\leq1$ for all $j',l'\in\K_{\iota}$; for each $\iota\in\N_{0,s}$. We note that in many applications, when a large $k$ is considered, the difference of one between the load of the workers within the same subset $\K_\iota$, can be made insignificant. These integer load differences can be made relatively small, if the work is divided into many small units; e.g., by increasing $k$. In the context of computing gradients, this is often the case when it is taken over a large dataset.

In order to appropriately define a close to balanced assignment, in the case where $(s+1)\nmid n$, we use the following definition of $d_s(\Bb)$; which gives a measure of how far from perfectly balanced the assignments of $\Bb$; in terms of the bound \eqref{bound_TLDK}.

\begin{Def}
\label{def_unif}
Define $d_s(\Bb)\coloneqq\sum\limits_{i=1}^n\left|\|\Bb_{(i)}\|_0-\frac{k}{n}(s+1)\right|$ for $\Bb\in\Z^{n\times k}$. This function measures how far the task allocations $\left\{\|\Bb_{(i)}\|_0\right\}_{i=1}^n$ are from being \textbf{uniform}, i.e., $\|\Bb_{(i)}\|_0= \floor{\frac{k}{n}(s+1)+\frac{1}{2}}$ for all $i\in\N_n$. Furthermore, $\{\|\Bb_{(i)}\|_0\}_{i=1}^n$ is uniform; i.e., all elements are equal, if and only if $d_s(\Bb)=0$.
\end{Def}

The objective of our proposed approach is to then give a binary solution to the optimization problem with rational constraints:
\begin{equation*}
\begin{aligned}
\OPR \qquad \arg\min_{\Bb\in\Q^{n\times k}} \quad & \big\{d_s(\Bb)\big\}\\
\textrm{s.t.} \ \ & \nnz(\Bb)=k\cdot(s+1)\\
& \bigsqcup_{i=0}^{s}\K_i=\N_n \ : \ \big||\K_j|-|\K_l|\big|\leq1, \ \forall j,l\in\N_{0,s} \\
& \big|\|\Bb_{(j)}\|_0-\|\Bb_{(l)}\|_0\big|\leq 1, \ \forall j,l\in \K_i, \ \forall i\in\N_{0,s}\\
& \exists\ab_\I\in\Q^n \text{ w/ } \supp(\ab_\I)\subseteq\I; \ \forall \I\in\I_f^n \ : \ \ab_\I^T\Bb=\bold{1}_{1\times k}
\end{aligned}\ ,
\end{equation*}
whose solutions yield \textit{almost} perfectly balanced GC schemes, with minimum total computational load across the network. The first constraint corresponds to the minimal total load of a GCS that is resilient to $s$ stragglers; from Lemma \ref{min_load_lem}. In order to have approximately equal cardinality across all partitions of the workers, we impose the second constraint; which is required by generalized fractional repetition coded computing schemes (Definition \ref{gen_FRC_def}). The third ensures that there is almost perfect balance among workers within the same partition $\K_i$; for each $i\in\N_{0,s}$. Together, the second and third constraints ensure that the maximum load across all workers is minimized, for the parameters $n$ and $s$. The fourth constraint imposes condition \eqref{GC_condition}, to guarantee that $\Bb$ is a generator matrix of a valid GCS.

When incorporating the additional constraint that the encoding-decoding pair is binary, we get the following binary integer program:
\begin{equation*}
\begin{aligned}
\IPB \qquad \arg\min_{\Bb\in\{0,1\}^{n\times k}} \quad & \big\{d_s(\Bb)\big\}\\
\textrm{s.t.} \ \ & \bigsqcup_{i=0}^{s}\K_i=\N_{0,n-1} \ : \ \big||\K_j|-|\K_l|\big|\leq1, \ \forall j,l\in\N_{0,s} \\
& \big|\|\Bb_{(j)}\|_0-\|\Bb_{(l)}\|_0\big|\leq 1, \ \forall j,l\in \K_i, \ \forall i\in\N_{0,s}\\
& \sum_{j\in\K_i}\Bb_{(j)}=\bold{1}_{1\times k}, \ \forall i\in\N_{0,s}
\end{aligned}\ .
\end{equation*}
The reduction from $\OPR$ to $\IPB$ when requiring that $\Bb$ and $\ab_\I$ are over $\{0,1\}$, is shown in the proof of Theorem \ref{thm_optimality} in Appendix \ref{app_sect_4}.

\subsection{Minimum Maximum Load of Workers in a Binary FRC Scheme}
\label{min_max_load_subsec}

Next, we show what the minimum maximum load over all workers of a binary GCS is, in order to construct a valid GCS for any pair of parameters $(s,n)$. Let $n=\ell\cdot(s+1)+r$ with $\ell=\floor{\frac{n}{s+1}}$. Note that $r\equiv n\bmod(s+1)$. Similarly, let $r=t\cdot\ell+q$; which specifies the Euclidean division of $r$ by $l$. Therefore, $n=\ell\cdot(s+t+1)+q$. In a particular case, we will also need the parameters specified by the division of $n$ by $(\ell+1)$. Let $n=\lambda\cdot(\ell+1)+\rt$ (if $\ell=s-r$, then $\lambda=s$). For clarity, similar to the work of \cite{TLDK17} which we are extending, we assume that $n=k$ for our main analysis. To complement this, in Subsection \ref{close_bal_des_subsec}, we explain our approach to handling the cases where $n\neq k$. The discussion in the aforementioned subsection also serves as a summary of our proposed encoding algorithm. To summarize, we have
\begin{equation}
\label{eq_1}
  n=\ell\cdot(s+1)+r \qquad \ 0\leq r<s+1
\end{equation}
\begin{equation}
\label{eq_2}
  r=t\cdot\ell+q \qquad \qquad \ \ \ 0\leq q<\ell \ \ \ind
\end{equation}
\begin{equation}
\label{eq_3}
  n=\lambda\cdot(\ell+1)+\rt \qquad \ 0\leq \rt<\ell+1
\end{equation}
where all terms are nonnegative integers.

We have already established that in order to tolerate $s$ stragglers; and not more than $s$ in the worst case, the workers need to be partitioned into $s+1$ subsets. In order for each of the subsets to have approximately the same size, for the partitioning $\bigsqcup_{i=0}^{s}\K_i=\N_n$, we assign each $\K_i$ either $\floor{\frac{n}{s+1}}$ or $\ceil{\frac{n}{s+1}}$ workers, i.e., $|\K_i|=\ell$ or $|\K_i|=\ell+1$ for all $i\in\N_{0,s}$. Note that on average, in partitions comprised of fewer workers, we need to allocate more data partitions to each worker; in order to compensate for the fact that fewer workers are needed to collectively compute all partial gradients $\{g_j\}_{j=1}^k$. Consequently, we do not want to assign less that $\ell$ workers to any partition $\K_i$. This observation corresponds to the second constraint of $\OPR$.

For what will follow, we set the partitions as
\begin{equation}
\label{congr_classes_Ki}
  \K_i=\big\{i+z\cdot (s+1):z\in\N_0\big\}\bigcap\N_{0,n-1}
\end{equation}
for each $i\in\N_{0,s}$, hence $|\K_\iota|=\ell+1$ for $\iota\in\{0,1,\ldots,r-1\}$ and $|\K_j|=\ell$ for $j\in\{r,\ldots,s\}$.\footnote{These are precisely the \textit{congruence classes} we will be referring to later.} By what was discussed above, on average; the workers corresponding to the subsets $\{\K_j\}_{j=r}^s$ will be allocated greater loads. This setup leads to the following statements. We first prove a simple proposition, which conveys the main idea behind our approach, and then provide a theorem with a precise lower bound on the load needed by the worker in the fractional repetition GCS who carries the maximum computation load, according to the parameters determined by \eqref{eq_1} and \eqref{eq_2}. This bound is achieved through our proposed scheme.

\begin{Prop}
\label{prop_mim_max_load}
Considering any binary generalized fractional repetition GCS of $n$ workers which tolerates $s$ stragglers, the minimum maximum load of any worker is $\big\lceil k/\floor{\frac{n}{s+1}} \big\rceil$. Specifically, given a generalized FRC encoding matrix $\Bb\in\{0,1\}^{n\times k}$, there always exists an $i\in\N_n$ for which $\|\Bb_{(i)}\|_0 \geq \big\lceil k/\floor{\frac{n}{s+1}} \big\rceil$.
\end{Prop}

\begin{proof}
By \eqref{eq_1}, to have a close to balanced partitioning according to the second constraint of $\OPR$, all worker partitions need to satisfy $|\K_i|\geq\floor{\frac{n}{s+1}}$, i.e., $\ell\geq\floor{\frac{n}{s+1}}$ where $\ell$ is precisely the minimum cardinality among all $\{\K_i\}_{i=0}^{s}$. Achievability follows from the assignment described above, where $|\K_\iota|=\ell+1$ for $\iota\in\{0,1,\ldots,r-1\}$ and $|\K_j|=\ell$ for $j\in\{r,\ldots,s\}$. For all worker partitions with cardinality $\ell$, it then follows that the maximum computational load among the workers of these subgroups, is at least $\ceil{\frac{k}{\ell}}$.

Now, any other worker partition which does not have cardinality $\ell$, has cardinality $\ell+1$. Following the same argument as above, the maximum computational load among the workers of these subgroups, is at least $\ceil{\frac{k}{\ell+1}}$. Since $\ceil{\frac{k}{\ell}} \geq \ceil{\frac{k}{\ell+1}}$, the binary fractional repetition GCS encoding matrix $\Bb$, satisfies
\begin{equation*}
\label{max_min_load_bd}
  \min \arg\max_{i\in\N_n} \quad \big\{\|\Bb_{(i)}\|_0\big\} \geq \Big\lceil k\big/\big\lfloor\frac{n}{s+1}\big\rfloor \Big\rceil\ .
\end{equation*}
\end{proof}

\begin{Thm}
\label{thm_mim_max_load}
Considering binary generalized FRC schemes of $n$ workers which tolerates $s$ stragglers; with $k=n$, the minimum maximum load of any worker is at most $s+t+2$. Specifically, given a FRC encoding matrix $\Bb\in\{0,1\}^{n\times n}$, there always exists an $i\in\N_n$ for which $\|\Bb_{(i)}\|_0\geq s+t+1$ when $q=0$; and $\|\Bb_{(i)}\|_0\geq s+t+2$ when $q>0$.
\end{Thm}

\begin{proof}
Consider a partition $\K_l$, for which $L=|\K_l|$ such that $L\geq\ell+1$. By Corollary \ref{eq_formulation_cor}, it follows that $\sum_{l'\in\K_l}\|\Bb_{(l')}\|_0=k$, and the allocation within the subset of workers indexed by $\K_l$ can be done so that each row indexed by $l'\in\K_l$ has support of size $\floor{\frac{k}{L}}$ or $\ceil{\frac{k}{L}}=\floor{\frac{k}{L}}+1$, for which
\begin{equation}
  \bfloor{\frac{k}{L}} \leq \bceil{\frac{k}{L}} \leq \bceil{\frac{k}{\ell+1}} = \bceil{\frac{n}{\ell+1}} \leq \bceil{\frac{n}{\ell}} \leq s+2
\end{equation}
i.e., $\|\Bb_{(l')}\|_0\leq s+t+2$ for all $l'\in\K_l$.

This shows that it suffices to consider the worker partitions of smaller sizes --- subsets $\{\K_j\}_{j=r}^s$; which have cardinality $|\K_j|=\ell<\ell+1$. Without loss of generality, we consider the partitioning \eqref{congr_classes_Ki}. For this partitioning, the $n^{th}$ worker is now assigned the index $0$, i.e., $\Bb_{(0)}$ corresponds to $\Bb_{(n)}$.

For a contradiction, assume that $\|\Bb_{(j')}\|_0\lneq s+t+2$ for all $j'\in\K_j$, in the case where $q>0$. It then follows that
\begin{equation}
  \sum_{i'\in\K_j}\|\Bb_{(j')}\|_0 \leq \ell\cdot(s+t+1) = \ell\cdot(s+1)+\ell\cdot t \lneq \ell\cdot(s+1)+\ell\cdot t+q = k
\end{equation}
where the last equality follows from \eqref{eq_1}, \eqref{eq_2} and $n=k$. Since $q>0$, we conclude that $\sum_{i'\in\K_j}\|\Bb_{(j')}\|_0<k$, which contradicts Corollary \ref{eq_formulation_cor}. Hence, there is at least one $j'\in\K_j$ for which $\|\Bb_{(j')}\|_0\geq s+t+2$.

Similarly, assume that $\|\Bb_{(j')}\|_0\lneq s+t+1$ for all $j'\in\K_j$, in the case where $q=0$ and $\ell>0$. It then follows that
\begin{equation}
  \sum_{i'\in\K_j}\|\Bb_{(j')}\|_0 \leq \ell\cdot(s+t) \lneq \ell\cdot(s+t)+\ell = \ell\cdot(s+1)+\ell\cdot t = k
\end{equation}
which contradicts Corollary \ref{eq_formulation_cor}. This completes the proof.
\end{proof}

It is worth noting that in the case where $(s+1)\mid n$; we have $r=t=q=0$, and Theorem \ref{thm_mim_max_load} reduces to bound \eqref{bound_TLDK} derived in \cite{TLDK17} for perfectly balanced schemes.

\section{Proposed Binary Gradient Coding Scheme}
\label{proposed_scheme_section}

In this section, we present our proposed binary coding scheme, along with its main properties. First, we present the construction of the encoding matrix $\Bb$; in Subsections \ref{bin_gc_sec}, \ref{subsec_B_C1} and \ref{subsec_B_C2}. In Subsection \ref{dec_vector_sec} we present an efficient online construction of the corresponding decoding vector $\ab_\I$. In Subsections \ref{val_opt_subsec} and \ref{n_gr_s2_subsec} we show that our proposed $\Bb$ and certain variants of it are as close to being uniform as possible; according to $\OPR$, and that in the regime $n\geq s^2$ the gap from being perfectly balanced is negligible. Finally, in Subsection \ref{heter_case_sec} we provide an analysis which determines how to appropriately allocate the assignments when the workers are heterogeneous, so that they have the same expected completion time --- this relaxes the close to uniform assignment.

\subsection{Close to Balanced Encoding Design}
\label{close_bal_des_subsec}

In this subsection, we give the approach of our proposed Algorithms \ref{B_ones_unb_unif_C1} and \ref{B_ones_unb_unif_C2} for the general cases where $n\neq k$. This also serves as a summary of Algorithms \ref{B_ones_unb_unif_C1} and \ref{B_ones_unb_unif_C2}. The decoding method, provided in Algorithm \ref{determine_aI}, remains the same for the cases where $n\neq k$. Note that Step 1 in which we partition $\N_n$ corresponds to the second constraint of $\OPR$, and Steps 2 and 3 in which we partition $\N_k$ according to the grouping of the worker subgroups corresponds to the third constraint of $\OPR$. In Steps 2 and 3, we assign an approximately uniform support to the rows of the GC encoding matrix $\Bb$.\\

\noindent \textbf{Step 1:} Partition the index set $\N_n$ of the $n$ workers into $s+1$ disjoint subgroups $\{\K_i\}_{i=0}^s$. According to \eqref{eq_1}, $r$ of these subgroups have cardinality of size $\ell+1$, while the remaining $s+1-r$ subgroups have cardinality of size $\ell$.\\

\noindent \textbf{Step 2:} Define $\zeta_1 \coloneqq \floor{\frac{k}{\ell+1}} = \Big\lfloor k\big/\big(\big\lfloor\frac{n}{s+1}\big\rfloor+1\big) \Big\rfloor$. We then assign to the $r$ groups with lighter loads, i.e., index subgroups with cardinality $\ell+1$, a load of $\zeta_1$ or $\zeta_1+1$.\\

\noindent \textbf{Step 3:} Define $\zeta_2 \coloneqq \floor{\frac{k}{\ell}} = \Big\lfloor k\big/\big\lfloor\frac{n}{s+1}\big\rfloor \Big\rfloor$. We then assign to the $s+1-r$ groups with heavier loads, i.e., index subgroups with cardinality $\ell$, a load of $\zeta_1$ or $\zeta_1+1$.\\

\begin{Rmk}
The approach through Steps 1-3, results in a close to balanced encoding matrix. Analogous to Theorem \ref{thm_mim_max_load}, we have $\zeta_1\leq\|\Bb_{(i)}\|_0\leq\zeta_2+1$, hence, the maximum load difference between any two workers is at most $\zeta_2-\zeta_1+1$. Furthermore, compared to bound \eqref{bound_TLDK}, the workers determined by Step 2 have an assignment load lower bound of $\zeta_1\geq\bfloor{\frac{k}{n+s+1}(s+1)}$, and those determined by Step 3 a lower bound of $\zeta_2\geq\bfloor{\frac{k}{n}(s+1)}$. It is worth noting that the latter lower bound is the floor of the bound determined in \cite{TLDK17}, while the former has the additive term of $s+1$ in the denominator; which is an artifact of the fact that we are considering a binary encoding matrix which is as sparse as possible. Furthermore, in practice, we would not have $k$ being significantly larger than $\ell$, which implies that $\zeta_2\approx\zeta_1$.
\end{Rmk}

To simplify the analysis that will follow, and draw direct comparisons to the work of \cite{TLDK17}, we will be assuming that $n=k$. This can always be met by the central server, who partitions $\D$ that is of size $N$; before distributing the parts $\{\D_j\}_{j=1}^k$. Since $N$ and $n$ are fixed, the central server can partition $\D$ according to the following procedure:
\begin{enumerate}[label=(\Roman*)]
  \item set $k=n$, $\taut = \floor{\frac{N}{n}}$, and $r'\equiv N\mod k$,
  \item partition $\D=\bigsqcup_{j=1}^{k} \D_j$ s.t.: $\{\D_j\}_{j=1}^{r'}$ are of size $\taut+1$, and $\{\D_j\}_{j=r'+1}^{k}$ are of size $\taut$,
\end{enumerate}
which guarantees that any two parts either have the same number of data points; or one of them has at most one additional point.

It is straightforward to verify that (I)-(II) partition the entire dataset $\D$ as desired. Since we have $r'$ parts of size $\taut+1$ and $k-r'$ parts of size $\taut$, it follows that
\begin{equation*}
  (\taut+1)\cdot r'+\taut\cdot(k-r') = \taut\cdot(r'+k-r')+r' = \taut\cdot k + r' = \taut\cdot n + r' = N
\end{equation*}
hence all $N$ data points of $\D$ are accounted for, only once in $\bigcup_{j=1}^k\D_j$. The assumption that $n=k$ is therefore not restrictive, as it can be incorporated into the coded computing scheme.

\subsection{Encoding Matrix}
\label{bin_gc_sec}

The idea is to work with congruence classes $\bmod(s+1)$ on the set of the workers' indices $\N_n$, in such a way that the workers composing a congruence class are roughly assigned the same number of partitions (differing by no more than one), while all partitions appear exactly once in each class. By \textit{congruence class} we simply mean the set of integers $j\in\N_n$ which are equivalent $\bmod\ (s+1)$. The classes are denoted by $\left\{[i]_{s+1}\right\}_{i=0}^s$. One could use a random assignment once it is decided how many partitions are allocated to each worker. However, in order to get a deterministic encoding matrix, we assign the partitions in ``blocks'', i.e., submatrices consisting of only $1$'s. With this setup, condition \eqref{cond_gen_GC_B} becomes
\begin{equation}
\label{sum_cond_BGC}
  \sum\limits_{j\equiv c\bmod(s+1)}\Bb_{(j)}=\bold{1}_{1\times k}
\end{equation}
for each $c\in\N_{0,s}$. As in the proof of Theorem \ref{thm_mim_max_load}, we reassign the index of the $n^{th}$ worker; so that $\Bb_{(0)}$ corresponds to $\Bb_{(n)}$.
 
In the proposed GCS, the encoding is identical for the classes $\Cfr_1\coloneqq\left\{[i]_{s+1}\right\}_{i=0}^{r-1}$, and is also identical for the classes $\Cfr_2\coloneqq\left\{[i]_{s+1}\right\}_{i=r}^{s}$. The objective is to design $\Bb$ to be as close to a block diagonal matrix as possible, and  we do so by ensuring that the difference in the load assignments between any two servers within the same set of classes $\Cfr_1$ or $\Cfr_2$; is at most one. We refer to the $l$ disjoint sets of consecutive $s+1$ rows of $\Bb$ as \textit{blocks}, and the submatrix comprised of the last $r$ rows as the \textit{remainder block}. Note that in total we have $\ell+1$ blocks, including the remainder block, and that each of the first $\ell$ blocks have workers with indices forming a complete residue system. We will present the task assignments for $\Cfr_1$ and $\Cfr_2$ separately. A numerical example where $n=k=11$ and $s=3$, is presented in Appendix \ref{example_app}.

\subsection{Repetition Assignment for Classes $0$ to $r-1$}
\label{subsec_B_C1}

In our construction each of the first $r$ residue classes also have an assigned row in the remainder block, such that we could assign $r$ partitions to the last worker of each class in $\Cfr_1$, and evenly assign $s+1$ to all other workers corresponding to $\Cfr_1$. Our objective though is to distribute the $k$ tasks among the workers corresponding to the $\ell+1$ blocks as evenly as possible, for the congruence classes corresponding to $\Cfr_1$, in such a way that \textit{homogeneous} workers have similar loads. By homogeneous we mean the workers have the same computational power, which implies that they exhibit independent and identically distributed statistics for the computation time of similar tasks.

Note that $n=(\ell+1)\cdot s+(\ell+r-s)$. Hence, when $\ell>s-r$, we can assign $s+1$ tasks to each worker in the first $\ell+r-s$ blocks, and $s$ tasks to the workers in the remaining $s+1-r$ blocks. In the case where $\ell\leq s-r$, we assign $\lambda+1$ tasks to the first $\rt$ blocks and $\lambda$ tasks to the remaining $\ell+1-\rt$ blocks. It is worth pointing out that $\lambda=s$ and $\rt=0$ when $\ell=s-r$, which means that every worker corresponding to $\Cfr_1$ is assigned $\lambda=s$ tasks, as $n=(\ell+1)\cdot s$.

A pseudocode for this encoding process is presented in Algorithm \ref{B_ones_unb_unif_C1}, in Appendix \ref{appendix_Enc_matr}. For coherence, we index the rows by $i$ starting from 0, and the columns by $j$ starting from $1$. We point out that when $\ell>s-r$, we have $\lambda=s$ and $\rt=\ell+r-s>0$. In the case where $\ell\leq s-r$, we need to invoke \eqref{eq_3} which was introduced solely for this purpose, as we need the remainder $\rt$ to be nonnegative. It follows that Algorithm \ref{B_ones_unb_unif_C1} can be reduced to only include the \textbf{else if} statement; eliminating the conditional clause.

\subsection{Repetition Assignment for Classes $r$ to $s$}
\label{subsec_B_C2}

For the workers corresponding to $\Cfr_2$, we first check if $q=0$. If this is the case, the $n$ partitions are evenly distributed between these workers, i.e., each worker is assigned $(s+t+1)$ partitions; as $n=\ell\cdot(s+t+1)$ and here we are only considering $\ell$ blocks. When $0<q<r$, we assign $(s+t+2)$ tasks to each worker of $\Cfr_2$ in the first $q$ blocks, and $(s+t+1)$ to the workers in the remaining $\ell-q$ blocks. A pseudocode for the encoding process is provided in Algorithm \ref{B_ones_unb_unif_C2}, in Appendix \ref{appendix_Enc_matr}.

The final step is to \textit{combine} the encodings of the classes $\Cfr_1$ and $\Cfr_2$ to get $\Bb$. That is, we combine the outcomes $\tilde{\Bb}_{\Cfr_1}$ of Algorithm \ref{B_ones_unb_unif_C1} and $\tilde{\Bb}_{\Cfr_2}$ of Algorithm \ref{B_ones_unb_unif_C2}. One could merge the two algorithms into one or run them separately, to get $\Bb=\tilde{\Bb}_{\Cfr_1}+\tilde{\Bb}_{\Cfr_2}$.

The encoding matrix $\Bb$ is also the adjacency matrix of a bipartite graph $G=(\mathcal{L},\mathcal{R},\mathcal{E})$, where the vertices $\mathcal{L}$ and $\mathcal{R}$ correspond to the $n$ workers and the $k$ partitions, respectively. We can also vary the number of stragglers $s$ the scheme can tolerate for a fixed $n$, by trading the sparsity of $\Bb$. In other words, if $\Bb$ is designed to tolerate more stragglers, then more overall partial gradients need to be computed. This results in more computations over the network, as $|\supp(\Bb)|=k\cdot(s+1)$. Lastly, the decoding step in our scenario and any binary FRC scheme, corresponds to a largest partial matching of the subgraph $G'=(\mathcal{L}',\mathcal{R},\mathcal{E})$ of $G$, where $\mathcal{L}'$ is the vertex subset corresponding to the $n-s$ fastest workers at that given iteration. 

\subsection{Decoding Vector}
\label{dec_vector_sec}

A drawback of the binary GCS introduced in \cite{TLDK17} is that it requires solving a system of linear equations to compute $\Ab\in\R^{{{n}\choose{f}}\times n}$; which contains the decoding vectors corresponding to all possible index sets of responsive workers $\I\in\I_f^n$. Specifically, for each possible $\I$, it computes $\ab_\I\coloneqq \bold{1}_{1\times f}\cdot\Bb_{(\I)}^\dagger$, where $\Bb_{(\I)}\in\R^{f\times k}$ is the restriction of $\Bb$ to the $f$ rows of $\Bb$ corresponding to $\I$ and $\dagger$ denotes the pseudoinverse operation. Computing $\Bb_{(\I)}^\dagger$ for $k>f$ without fast matrix-matrix multiplication, requires $O(2kf^2+f^3)=O(f^3)$ operations, which is done ${{n}\choose{f}}$ distinct times. Then, matrix $\Ab$ needs to be stored and searched through at each iteration of the gradient descent procedure. Searching through $\Ab$ to find the corresponding $\ab_{\I}$ is prohibitive; as it has $\Theta(n^s)$ rows to look through. We propose a more efficient online decoding algorithm in order to mitigate this problem.

The construction of the decoding matrix in \cite{TLDK17} when using regular matrix multiplication and inversion, requires $O\left(k^3(k+2n-2s)\right)$ operations to construction the pseudoinverse of a submatrix of $\Bb$, for each of the possible ${n}\choose{s}$ index sets $\I$. Furthermore, at each iteration of gradient descent, the decoding step requires a search through the rows of $\Ab$. This comes at a high cost when compared to our online algorithm, which constructs a decoding vector in $O(n+s)$ operations, and does not require any additional storage space.

We point out that a similar decoding approach was developed independently in \cite{CP18}, which focuses on approximating the gradient rather than recovering the exact gradient. The main idea behind the two approaches is that we look at the index set of responsive workers, and then only consider the response of workers with mutually exclusive assigned partitions. While the objective of the scheme in \cite{CP18} is to form a vector which is close to $\bold{1}_{1\times k}$, we guarantee that this vector is attained once $n-s$ workers respond. By \eqref{GC_identity}, we can therefore recover the exact gradient.

In the proposed binary GCS there is no rescaling of the partial gradients taking place by encoding through $\Bb$, as the coefficients are $1$ or $0$. As a result, the proposed decoding reduces to simply summing a certain subset of the completed encoded tasks, while making sure that each partial gradient is added exactly once. To this end, among any $f$ workers who send back their computed sum of partial gradients, we need to have $\ell$ workers, $\ell=\frac{n}{s+1}\in\Z_+$ (or $\ell+1$ where $\ell=\floor{\frac{n}{s+1}}$, if $(s+1)\nmid n$), who have no common assigned partitions. We elaborate on this in the next paragraph.

If $r=0$, the decoder traverses through the $s+1$ classes consecutively to find one which is a complete residue system (Algorithm \ref{determine_aI}). This will be used to determine the decoding vector $\ab_{\I}$, implied by Corollary \ref{eq_formulation_cor}. When $r>0$, the decoder first traverses through the last $s+1-r$ congruence classes; checking only the first $\ell$ blocks. If it cannot find a complete residue system corresponding to returned tasks by non-stragglers, it proceeds to the first $r$ classes; checking also the remainder block. This extra step makes the scheme more efficient. In both cases, by the pigeonhole principle we are guaranteed to have a complete residue system, provided that $f$ completed tasks are received.

The next step is to devise a decoding vector for each of the ${n}\choose{s}$ different straggler scenarios $\I$. We associate the $i^{th}$ complete residue class with a decoding vector $\ab_i$ defined as
\begin{equation}
\label{vecs_ai}
\ab_i\coloneqq \sum_{j\in[i]_{\ell}}\eb_j \ \in\{0,1\}^{n},
\end{equation}
for $i\in\N_{0,\ell-1}$. Also, note that $\|\ab_i\|_{0}=\ell+1$ for the decoding vectors corresponding to the first $r$ classes, and $\|\ab_i\|_{0}=\ell$ for the remaining classes. In both cases, $\ab_{i+1}$ is a cyclic shift of $\ab_i$.

At each iteration the gradient is computed once $f$ worker tasks are received. Define the \textit{received indicator-vectors}
\begin{equation}
  \left(\text{rec}_{\I}\right)_i = \begin{cases} 1 \qquad \text{ if } i\in\I \\ 0 \qquad \text{ if } i\not\in\I \end{cases},
\end{equation}
for each possible $\I$, where $\|\text{rec}_{\I}\|_0=f$ and $n-\|\text{rec}_{\I}\|_0=s$. Thus, there is at least one $i\in\N_{0,\ell-1}$ for which $\supp(\ab_i)\subsetneq \supp(\text{rec}_{\I})$. If there are multiple $\ab_i$'s satisfying this property, any of them can be selected. The pseudocode is presented in Algorithm \ref{determine_aI}.

\begin{algorithm}[h]
\label{determine_aI}
\SetAlgoLined
\KwIn{received indicator-vector $\text{rec}_{\I}$}
\KwOut{decoding vector $\ab_{\I}$}
 \For{$i = s$ to $0$}
  {
   \If {$\left(\mathrm{rec}_{\I}\right)_i \equiv 1$}
    {
     $l\gets i$ \\
     \If {$\supp(\ab_l)\subseteq \supp(\mathrm{rec}_{\I})$}
     {
      $\ab\gets \ab_l$ \Comment{$\ab_l$ is defined in \eqref{vecs_ai}}\\
      break
     }
    }
  }
 \Return $\ab_{\I}\gets\ab$
 \caption{Decoding Vector $\ab_{\I}$}
\end{algorithm}

\subsection{Validity and Optimality of our GCS}
\label{val_opt_subsec}

Now that we have presented our construction, we provide the accompanying guarantees in terms of validity, and optimality, which motivated our construction.

\begin{Thm}
\label{thm_GC}
The proposed encoding-decoding pair $(\Bb,\ab_\I)$ satisfy condition \eqref{GC_condition}, for any $\I\in\I_f^n$. That is, $(\Bb,\ab_\I)$ comprise a valid GCS which tolerates up to $s$ stragglers.
\end{Thm}

\begin{proof}
By our construction of $\Bb$, the rows corresponding to a congruence class are mutually exclusive and their superposition is precisely $\bold{1}_{1\times k}$.

By the pigeonhole principle the number of completed encoded tasks that is required at the decoder to guarantee a successful recovery of the gradient, denoted by $\nu$, is equal to
\begin{align}
\label{expr_nu}
  \nu &\coloneqq \ell\cdot r+(\ell-1)\cdot\big[(s+1)-r\big]+1\nonumber\\
  &= \ell\cdot(s+1)-s+r = \big[\ell\cdot(s+1)+r\big]-s=n-s\ .
\end{align}
Therefore, as long as $\nu=n-s$ many workers respond, there is at least one subset of them whose indices form a complete residue system $\bmod(s+1)$. Algorithm \ref{determine_aI} determines such an index subset, and constructs a binary vector whose support corresponds to this subset. As a result, $\ab_\I^T\Bb=\bold{1}_{1\times k}$ for any $\I$ of size $n-s$, and \eqref{GC_identity} is satisfied.
\end{proof}

Note that the total number of task assignments is $k\cdot (s+1)$, for any pair of integers $(s,n)$ where $0\leq s<n$, as expected. This is the same total load required by the MDS based schemes. Also, our GCS meets the lower bound on total task assignments of $\Bb$ implied by \eqref{bound_TLDK} and Lemma \ref{min_load_lem}
\begin{equation}
\label{bound_TLDK_B}
  \nnz(\Bb)=\sum_{i=1}^n\|\Bb_{(i)}\|_0 \geq n\cdot\frac{k}{n}(s+1)=k\cdot(s+1)
\end{equation}
with equality.

It can be observed that the runtime complexity of Algorithm \ref{determine_aI} is  $O((\ell+1)\cdot(s+1))=O(n+s)$. This complexity can be slightly reduced by the following modification. The for-loop in Algorithm \ref{determine_aI} can be stopped early by only traversing through the classes $0,\cdots,s-1$, and assigning $\ab_{\I}\gets\ab_{s}$ if none was selected. This reduces the runtime complexity to $O((\ell+1)\cdot s)$, hence our proposed decoder is significantly faster than the decoding algorithm of \cite{TLDK17}. The decoding matrix $\Ab$ of \cite{TLDK17} requires ${n}\choose{s}$ applications of a pseudoinverse for its construction, which makes it impractical for large ${n}\choose{s}$. Once this has been constructed, at each gradient descent iteration it requires an additional decoding step which involves searching through the ${n}\choose{s}$ rows of $\Ab$.

An alternative decoding is to consider a decoding of each of $\Cfr_1$ and $\Cfr_2$ separately; in a streaming fashion, and terminate whenever one of the two is completed. This decoding procedure will be especially useful in our first CMMS. Both the CMMS and the alternative decoding will be described in more detail in Subsection \ref{str_dec_subs}. As was done in \cite{TLDK17}, to simplify the presentation of the proof of Theorem \ref{thm_optimality}, we restrict our attention to the case when $n=k$.

\begin{Thm}
\label{thm_optimality}
Let $n=k$. The task allocation through $\Bb$ resulting from Algorithms \ref{B_ones_unb_unif_C1} and \ref{B_ones_unb_unif_C2} is a solution to the optimization problem $\OPR$; when we impose the restriction that $\Bb$ is binary. Specifically, $\Bb$ is a solution to $\IPB$.
\end{Thm}

Theorem \ref{thm_optimality} holds for permutations of the columns of $\Bb$, or a random assignment of tasks per class; as opposed to repeating blocks --- as long as all partitions are present only once in a single worker of each congruence class. The decoding in either of these cases remains the same. Furthermore, the proposed $\Bb$ can be viewed as an extension of the \textit{cyclic repetition} scheme introduced in \cite{TLDK17}. An example of how the allocations can be modified for each congruence class is given in Appendix \ref{example_app}. By ``valid permutation per congruence class'', we mean that a separate permutation is applied to the columns of $\Bb\big|_{[c]_{s+1}}$; the restriction of $\Bb$ to the rows corresponding to $[c]_{s+1}$, for each congruence class $c\in\N_{0,s}$.

For the purpose of the applications considered in this paper, permutations of the rows or columns of $\Bb$ do not affect the overall performance or guarantees of the proposed GCS; i.e., any permutation applied to the encoding matrix of the approach of Algorithms \ref{B_ones_unb_unif_C1} and \ref{B_ones_unb_unif_C2} would have the same result.\footnote{ In the case where a permutation is applied to the \textit{rows} of $\Bb$, the decoding vectors defined in \eqref{vecs_ai} should be modified accordingly.} A permutation of the rows corresponds to permutation of the workers' indices, and a permutation of the columns simply means the workers are assigned different partitions, but since the same number of partitions is allocated to each worker, their total workload remains the same.

In the following theorem, we provide the conditions for which an alternative binary encoding matrix $\bar{\Bb}$ can be used as an encoding matrix for a fractional repetition GCS, by comparing it to the resulting $\Bb$ of our Algorithms \ref{B_ones_unb_unif_C1} and \ref{B_ones_unb_unif_C2}. The encoding through $\bar{\Bb}$ is a valid permutation of $\Bb$, as was defined above.

\begin{Thm}
\label{cond_B_thm}
A binary encoding matrix $\bar{\Bb}$ is a valid permutation of the task allocations per congruence class of the encoding matrix $\Bb$ proposed by Algorithms \ref{B_ones_unb_unif_C1} and \ref{B_ones_unb_unif_C2}, for which Algorithm \ref{determine_aI} produces a correct decoding vector, i.e., $\ab_{\I}^T\bar{\Bb}=\bold{1}_{1\times k}$ for all possible $\I\in\I_f^n$, if:
\begin{enumerate}
  \item $\|\bar{\Bb}_{(i)}\|_0=\|\Bb_{(i)}\|_0$
  \item {\normalfont$\supp(\bar{\Bb}_{(i)})\bigcap\supp(\bar{\Bb}_{(j)})=\emptyset \quad $} if $\ \ i\equiv j\bmod(s+1)$
\end{enumerate}
for $i,j\in\N_n$ distinct. These conditions also imply that $\|\bar{\Bb}^{(i)}\|_0=s+1$ for all $i\in\N_k$.
\end{Thm}

\subsection{Distribution of Assignments for $n\geq s^2$}
\label{n_gr_s2_subsec}

Considering the identities \eqref{eq_1}, \eqref{eq_2} and \eqref{eq_3}, note that for $\ell>r$ we have $t=0$ and $r=q$. Furthermore, when $\ell=s$ we have $n=s\cdot(s+1)+r\approx s^2$, and in the regime $n\geq s^2$, we can show that $t$ is at most 1. Then, the gap between the heaviest and lightest loads; respectively $s+t+2$ and $s$, is at most 3.

\begin{Lemma}
Let $n=s^2+a$ for $a\in \N_0$ and $s<n$. Then, we have $t=1$ only when $a=s-2,s-1$ or $2s$. Otherwise, $t=0$.
\end{Lemma}

\begin{proof}
We break up the proof into three cases:
\vspace{3mm}

\noindent \underline{Case $a\in\{0,\cdots,s-3\}$}: For $\alpha = s-a\in\{3,4,\cdots,s\}$:
\begin{equation}
  n = s\cdot(s+1)-\alpha = \overbrace{(s-1)}^{\ell}\cdot(s+1)+\overbrace{(s+1-\alpha)}^r,
\end{equation}
and $\ell>r$ for any $\alpha$. Thus, $t=0$ and $r=q$.
\vspace{3mm}

\noindent \underline{Case $a\in\{s,\cdots,2s-1\}$}: Let $n=s^2+a=s^2+(s+\beta)$ for $\beta\in\{0,\cdots,s-1\}$. Then $n=s\cdot(s+1)+\beta$ implies $\ell=s$ and $r=\beta$. Since $r<\ell$, it follows that $t=0$ and $r=q$.
\vspace{3mm}

\noindent \underline{Case $a\gneq 2s$}: The final case to consider is $a=2s+\gamma$, for $\gamma\in\Z_+$. The resulting parameters are $r=q=\rem\big(\rem\big(\gamma,s+1\big)-1,s+1\big)$, $\ell=(s^2+2s+\gamma-r)/(s+1)$ and $t=0$.
\vspace{3mm}

When $\alpha=1$ it follows that $r=s$ and $\ell=s-1$, thus $t=1$ and $q=1$. For $\alpha=2$ we get $r=\ell=s-1$, hence $t=1$ and $q=0$. For both $\alpha=1$ and $\alpha=2$; $t=1$ is a consequence of $r\geq \ell$. In addition, when $\beta=s$ we have $r=\ell=s$; thus $t=1$ and $q=0$.
\end{proof}

\subsection{Task Allocation to Heterogeneous Workers}
\label{heter_case_sec}

We now discuss how to allocate the partitions when the workers are of \textit{heterogeneous} nature, such that all workers have the same expected execution time. This analysis may be needed in applications with very discrete, indivisible jobs. In cases where the work can be divided more finely, for any given subset of workers who are to share the total work, one can simply divide the work to be done among the workers in proportion to their computational strengths, to equalize the expected completion time. We present the case where we have two groups of machines, each consisting of the same type. By equalizing the expected completion time, we reduce the variance of the expected response times across the servers, for any possible set of $n-s$ non-straggling servers. The analysis for more than two groups of machines can be done in a similar fashion.

The two types of workers are denoted by $\T_i$; with a total of $\tau_i$ machines, and their expected execution for computing $g_j$ (for equipotent $\D_j$'s) by
\begin{equation}
  t_i\coloneqq\E\left[\text{time for } \T_i \text{ to compute } g_j\right],
\end{equation}
for $i\in\{1,2\}$, where $t_1\lneq t_2$; i.e., machines $\T_1$ are faster. Let $|\J_{\T_i}|$ denote the number of partitions each worker of $\T_i$ is assigned. The goal is to find $|\J_{\T_1}|$ and $|\J_{\T_2}|$ so that
\begin{equation}
\label{expecation_condition}
  \E\left[\T_1 \text{ compute their task}\right]=\E\left[\T_2 \text{ compute their task}\right],
\end{equation}
implying $t_1\cdot|\J_{\T_1}|=t_2\cdot|\J_{\T_2}|$. Hence $|\J_{\T_1}|\gneq|\J_{\T_2}|$, as $t_1\lneq t_2$. Let $\tau_1=\frac{\alpha}{\beta}\cdot\tau_2$ for $\frac{\alpha}{\beta}\in\Q_+$ in reduced form. Since $\tau_1+\tau_2=n$, it follows that
\begin{equation}
  \tau_1=\frac{\alpha}{\alpha+\beta} n \qquad \text{ and } \qquad \tau_2=\frac{\beta}{\alpha}\tau_1=\frac{\beta}{\alpha+\beta} n.
\end{equation}
To simplify the presentation of the task assignments, we assume $(s+1)\mid n$. If $(s+1)\nmid n$, one can follow a similar approach to that presented in Subsection \ref{bin_gc_sec} to obtain a close to uniform task allocation; while \textit{approximately} satisfying \eqref{expecation_condition}.

The main idea is to fully partition the data across the workers, such that each congruence class is comprised of roughly $\frac{\alpha}{\alpha+\beta}\cdot\frac{k}{s+1}$ workers of type $\T_1$, and $\frac{\beta}{\alpha+\beta}\cdot\frac{k}{s+1}$ workers of type $\T_2$. We want $\frac{\tau_1+\tau_2}{s+1}=\frac{n}{s+1}$ many workers for each congruence class, and
\begin{equation}
  |\J_{\T_1}|\cdot\frac{\tau_1}{s+1}+|\J_{\T_2}|\cdot\frac{\tau_2}{s+1}=k
\end{equation}
partitions to be assigned to each class. That is, the dataset $\D$ is completely distributed across each congruence class, and our GCS is designed accordingly.

Putting everything together, the following conditions determine $|\J_{\T_1}|$ and $|\J_{\T_2}|$
\begin{enumerate}[label=(\roman*)]
  \item $t_1\cdot|\J_{\T_1}|=t_2\cdot|\J_{\T_2}| \quad \iff \quad |\J_{\T_2}| = \frac{t_1}{t_2}\cdot|\J_{\T_2}|$
  \item $|\J_{\T_1}|\cdot\tau_1+|\J_{\T_2}|\cdot\tau_2=(s+1)\cdot k$
  \item $\tau_2=\frac{\beta}{\alpha}\cdot\tau_1 \quad \iff \quad \tau_1=\frac{\alpha}{\beta}\cdot\tau_2$.
\end{enumerate}
By substituting (iii) into (ii) to solve for $|\J_{\T_2}|$, and then plugging it into (i) to solve for $|\J_{\T_1}|$, we get
\begin{equation}
\label{num_jobs_1}
  |\J_{\T_1}|=(s+1)\cdot k \cdot \left(\frac{\alpha t_2}{\alpha t_2+\beta t_1}\right)\cdot\frac{1}{\tau_1}
\end{equation}
\begin{equation}
\label{num_jobs_2}
  |\J_{\T_2}|=(s+1)\cdot k \cdot \left(\frac{\beta t_1}{\alpha t_2+\beta t_1}\right)\cdot\frac{1}{\tau_2} 
\end{equation}
which we round to get appropriate numbers of assignments.

This framework may be generalized to any number of different groups of machines. Under the same assumptions, for $\T_1,\cdots,\T_m$ different groups with $t_i\lneq t_{i+1}$ for all $i\in\N_{m-1}$:
\begin{enumerate}[label=(\roman*)]
  \item $t_1\cdot|\J_{\T_1}|=t_2\cdot|\J_{\T_2}|=\cdots=t_m\cdot|\J_{\T_m}|$
  \item $|\J_{\T_1}|\cdot\tau_1+|+|\J_{\T_2}|\cdot\tau_2+\cdots+|\J_{\T_m}|\cdot\tau_m=(s+1)\cdot k$
  \item $\tau_1=\frac{\alpha_2}{\beta_2}\cdot\tau_2=\cdots=\frac{\alpha_m}{\beta_m}\cdot\tau_m$, for $\frac{\alpha_{i+1}}{\beta_{i+1}}\in\Q_+$
\end{enumerate}
need to be met. This gives us a system of $2(m-1)+1=2m-1$ equations with $m$ unknowns $\{|\J_{\T_j}|\}_{j=1}^m$, which is solvable.

\section{Binary Coded Matrix Multiplication Schemes}
\label{matr_mult_sec}

Multiplication of two matrices is one of the most common operations. Coded matrix multiplication is a principled framework for providing redundancy in distributed networks, to guarantee recovery in the presence of stragglers \cite{LSR17}. As in GC, each worker is requested to carry out some computation and encode it; before sending it back to the central server. In this section we first show how \textit{any} GCS can be used to devise a CMMS, and then present two different schemes based on our binary GCS. Like the fractional repetition scheme, our two schemes resemble replication codes. For simplicity in presentation, throughout this section we assume that $k\mid N$.

The two CMMSs have applications beyond matrix multiplication, which we discuss in Subsection \ref{comp_section}. Regarding matrix multiplication, the schemes have different trade-offs in terms of communication, storage, and computational operations, required by each worker. Depending on the application and the resources available, one may even be easier to implement compared to the other.

As pointed out in \cite{LKYSA18}, despite recent advancements in distributed gradient computations, schemes under parameters $(s,n)$ have not been developed which have a \textit{recovery threshold} (i.e., the worst case minimum number of workers that need to respond in order to recover the full gradient) less than $f=n-s$. On the other hand, many CMMSs exhibit considerably better recovery thresholds --- the optimum asymptotic recovery threshold of $\mu\nu$ for $1/\mu$ and $1/\nu$ respectively the fraction of $A$ and $B$ stored by each worker; was achieved through Polynomial Codes \cite{YMAA17}.

Improving the recovery threshold comes at the cost of trading encoding and decoding complexities, restrictions on how the matrices are partitioned, and storage. The two schemes we propose have a recovery threshold of $f=n-s$, though do not suffer from any of the aforementioned drawbacks. For simplicity in presentation, we assume that $N=\ell\cdot k$ for $\ell\in\Z_+$ and $N$ the effective dimension; which implies that we have a balanced assignment. When this is not the case, the analysis carried out in Subsection \ref{bin_gc_sec} can be applied.

\subsection{$\cmmO$ --- Outer-Product Representation}
\label{1st_sch_sec}

Consider a single central server node that has at its disposal the matrices $A\in\R^{L\times N}$ and $B\in\R^{N\times M}$, and it can distribute submatrices of $A$ and $B$ among $n$ workers; to compute their product $C=AB$ in an accelerated manner. One way of computing $C$ is to leverage the fact that
\begin{equation}
  C=\sum_{i=1}^NA^{(i)}B_{(i)}
\end{equation}
which has also been used in \cite{FJHDCG17,DFHJCG19}. Recall that $A^{(i)}$ denotes the $i^{th}$ column of $A$, and $B_{(i)}$ the $i^{th}$ row of $B$, as specified in Subsection \ref{notation_sec}. This makes the process parallelizable. To make use of this outer-product representation, we partition $A$ and $B$ each into $k$ disjoint submatrices consisting of $\tau=N/k$ columns and rows respectively, which we denote by $\At_j\in\R^{L\times \tau}$ and $\Bt_j\in\R^{\tau\times M}$ for $j=1,\cdots,k$. That is
\begin{equation}
\label{block_partition}
  A=\Big[\At_1 \ \cdots \ \At_k\Big] \quad \text{ and } \quad B=\Big[\Bt_1^T \ \cdots \ \Bt_k^T\Big]^T.
\end{equation}

The central server is then capable of distributing the pairs $(\At_j,\Bt_j)$ appropriately, with a certain level of redundancy, in order to recover $C$
\begin{equation}
  C=\sum_{j}^k\At_j\Bt_j\ .
\end{equation}
Define $X_j\coloneqq \At_j\Bt_j\in\R^{L\times M}$ for all $j$, and the matrix
\begin{equation}
\label{augm_matx_X}  
  \Xb \coloneqq \Big[X_1^T\ | \ \cdots\ | \ X_k^T\Big]^T \in \R^{kL\times M},
\end{equation}
similar to how $\gb$ was defined \eqref{g_matrix} in Section \ref{str_problem_GC}.
Recall that the main idea behind GC is to construct the pair of encoding matrix $\Bb$ and decoding vector $\ab_{\I}$, such that $\ab_{\I}^T\Bb=\bold{1}_{1\times k}$ for all ${{n}\choose{s}}$ possible index sets $\I$. A CMMS can be devised by the pair $(\Bb,\ab_{\I})$. The matrix product $C=AB$ is described as:
\begin{equation}
\label{GC_to_CMM}
  C = (\overbrace{\ab_{\I}^T\otimes\Ib_L}^{\abt_{\I}^T})\cdot((\overbrace{\Bb\otimes\Ib_L}^{\Bbt})\cdot\Xb) = \overbrace{(\underbrace{\bold{1}_{1\times k}}_{\ab_{\I}^T\Bb}\otimes\Ib_L)}^{\abt_\I^T\Bbt}\cdot\Xb = \sum_{j=1}^k X_j,
\end{equation}
where $\Bbt\in\C^{nL\times kL}$ is now the encoding matrix for the CMM, and $\abt_{\I}\in\C^{nL\times L}$ is the \textit{decoding matrix} corresponding to the non-straggler index set $\I$. Expression \eqref{GC_to_CMM} is analogous to \eqref{GC_identity}.

\begin{Thm}
\label{thm_CMM}
Any GCS can be extended to a coded matrix multiplication or addition scheme.
\end{Thm}

\begin{proof}
Consider a GCS $(\Bb,\ab_{\I})$, for which $\ab_{\I}^T\Bb=\bold{1}_{1\times k}$. By \eqref{GC_to_CMM} it follows that
\begin{equation}
\label{enc_dec_CMM}
  \abt_{\I}^T(\Bbt\Xb)=\sum_{j=1}^kX_j=C \ .
\end{equation}
Therefore, a CMM method $(\Bbt,\abt_{\I})$ is obtained.\\
$\ind$ For matrix addition, we simply construct $\Xb$ by augmenting the $k$ equi-sized matrices we want to add; instead of the products $\{\At_j\Bt_j\}_{j=1}^k$ in \eqref{augm_matx_X}, and we obtain a coded matrix addition scheme.
\end{proof}

Let $(\Bb,\ab_\I)$ be the encoding-decoding GC pair from Section \ref{BGC_sec}. In Theorem \ref{thm_CMM}, the resulting pair $(\Bbt,\abt_{\I})$ is a CMMS whose encoding matrix $\Bbt=\Bb\otimes\Ib_L$ represents the partition pairs $(\At_j,\Bt_j)$ as the columns of $\Bbt$; and its rows represent the $n$ workers. That is, the worker corresponding to the $i^{th}$ row of $\Bb$ receives the partition pairs corresponding to $\J_i=\supp(\Bb_{(i)})$, and is asked to send back the summation of the outer-products
\begin{equation}
\label{out_pro_identity}
  C_j\coloneqq\sum_{j\in\J_i}\At_j\Bt_j = \sum_{j\in\J_i}X_j\ .
\end{equation}

The decoding matrix $\abt_{\I}^T=\ab_{\I}^T\otimes\Ib_L$ only involves the computations of a complete residue system associated with the received workers, which are determined by $\supp(\ab_{\I})$.

The communication cost per worker which along with the storage required at the central server are the major drawbacks of this approach. Each worker will have to send back a matrix of size $L\times M$, and in the worst case, the central server will need to store $k\cdot(s+1)$ matrices of this size before it can recover $C$. The computation cost per worker is equivalent to that of multiplying two matrices, of dimensions corresponding to the block pairs. An alternative CMM decoding process overcomes the storage issue at the central server, which is described next.

\subsection{Decoding as a Streaming Process}
\label{str_dec_subs}

In the case where $(s+1)\mid n$, we can use a streaming process for the recovery of $C$. In this process, we only retain a single computation corresponding to each of the blocks of the encoding matrix $\Bb$; where $\Bb$ is now a block diagonal matrix with $\ell=\frac{n}{s+1}$ diagonal blocks of the form $\bold{1}_{(s+1)\times\floor{k/\ell}}$ or $\bold{1}_{(s+1)\times\ceil{k/\ell}}$. The process terminates once a single worker from each block has responded. The pseudocode for this procedure is given in Algorithm \ref{stream_dec}.

\begin{algorithm}[h]
\label{stream_dec}
\SetAlgoLined
\KwIn{computations $C_j$ sequentially}
\KwOut{product $C$}
\textbf{Initialize:} $C=\bold{0}_{L\times N}$, and $R=\emptyset$ the index set of the received workers' blocks\\
\While{$|R|<\ell$}
{
receive computation $C_j$\Comment{$j\in\N_n$}\\ 
$\hat{\ell}\gets \ceil{j/(s+1)}$\Comment{block index of the $j^{th}$ worker}\\  
 \If{$\hat{\ell}\notin R$}
  {
   $C\gets C+C_j$\\
   $R\gets R\cup\{\hat{\ell}\}$\\
  }
}
 \caption{Decoding in a Streaming Fashion}
\end{algorithm}

The benefit of this approach, compared to the decoding $\abt_{\I}^T=\ab_{\I}^T\otimes\Ib_L$ for $\ab_{\I}$ from Algorithm \ref{determine_aI}, is that the central server will never need to store more than double the entries of the product $C$. In the case where $(s+1)\nmid n$, we can do the exact same process by simply breaking the problem into two subroutines, one dealing with the workers whose indices correspond to the first $r$ congruence classes $\bmod(s+1)$, and the other with the workers corresponding to the remaining $s+1-r$ congruence classes. That is, we will work with $\Bb_{\Cfr_1}$ for $\ell+1$ blocks; and $\Bb_{\Cfr_2}$ for $\ell$ blocks separately. We carry out Algorithm \ref{stream_dec} in parallel for the two cases, and terminate whenever one of the two has computed $C$. Now, the central server will need to store a total number of entries no more than twice the size of matrix $C$. This decoding procedure can be done analogously for our GCS. An example with further details is provided in \ref{example_stream_dec}.

Algorithm \ref{stream_dec} takes into account \textit{which} workers have responded up to a certain instance, rather than only the \textit{total number} of workers which have responded. The recovery threshold in the worst case is $n-s$, matching that of our previous decoding procedure. On average though, considering all possible index sets $\I$ of responsive workers which correspond to a valid decoding according to Algorithm \ref{stream_dec}, less workers than the worst case of $n-s$ need to respond.

If $(s+1)\nmid n$, the worst case occurs when all workers corresponding to $\ell$ blocks of $\Bb_{\Cfr_1}$ and $\ell-1$ blocks of $\Bb_{\Cfr_2}$ respond, along with only one worker from either of the two remaining blocks. By \eqref{expr_nu}, the total number of responsive workers is $n-s$. In the best case, we need a single worker corresponding to each block of $\Bb_{\Cfr_2}$ to respond, i.e., $\ell=\floor{\frac{n}{s+1}}$ responsive workers. Similarly, if $(s+1)\mid n$, in the best case we require $\ell=\frac{n}{s+1}$ workers to respond, and in the worst case $n-(\ell-1)\cdot(s+1)+1=n-s$ many workers.

\subsection{$\cmmT$ --- Augmentation of Submatrices}
\label{2nd_sch_sec}

In a system where the main limitation is the communication load which can be handled from the workers to the central server; as well as storage of the computed task at the worker nodes, $\cmmO$ is not ideal, even with the more efficient decoding process. Next, we discuss an alternative CMMS which is superior in these aspects.

In contrast to the partitioning \eqref{block_partition} of $\cmmO$, in this scheme we partition $A$ along its rows and $B$ along its columns, as was done for the Polynomial codes in \cite{YMAA17}, i.e.,
\begin{equation}
\label{part_ccm2}
  A=\Big[\Abar_1^T \ \cdots \ \Abar_{k_1}^T\Big]^T \quad \text{ and } \quad B=\Big[\Bbar_1 \ \cdots \ \Bbar_{k_2}\Big],
\end{equation}
where $\Abar_j\in\R^{\frac{L}{k_1}\times N}$ and $\Bbar_j\in\R^{N\times \frac{M}{k_2}}$. Each worker computes the product of a submatrix of $A$ with a submatrix of $B$, and then the central server augments the received computations accordingly. 

For coherence, we let $k=k_1k_2$ for $k_1,k_2\in\Z_+$. To simplify the presentation of our scheme, we consider the case where $k_1|L$, $k_2|M$ and $(s+1)\mid k$; i.e., $\texS=L/k_1\in\Z_+$ and $\texT=M/k_2\in\Z_+$, and similar to our GCS that $n=k$. The product $C$ of the two matrices under this partitioning is equal to
$$ \begin{pmatrix} \boxed{\Abar_1\Bbar_1} & \boxed{\Abar_1\Bbar_2} & \hdots & \boxed{\Abar_1\Bbar_{k_2-1}}  & \boxed{\Abar_1\Bbar_{k_2}} \\ \boxed{\Abar_2\Bbar_1} & \ddots & & & \boxed{\Abar_2\Bbar_{k_2}} \\ \vdots & & \ddots & & \vdots \\ \boxed{\Abar_{k_1-1}\Bbar_1} & & & \ddots & \boxed{\Abar_{k_1-1}\Bbar_{k_2}} \\ \boxed{{\white1}\Abar_{k_1}\Bbar_1{\white1}} & \boxed{\Abar_{k_1}\Bbar_2} & \hdots & \boxed{ \Abar_{k_1}\Bbar_{k_2-1}} & \boxed{{\white1}\Abar_{k_1}\Bbar_{k_2}{\white1}} \end{pmatrix} $$
where we denote each product submatrix by $\Cbar_{i,j}\coloneqq\Abar_i\Bbar_j\in\R^{\texS\times\texT}$, and each block row column by
\begin{equation}
  \Cbar_i=\Big[\Cbar_{i,1}\ \cdots\ \Cbar_{i,k_2}\Big]\in\R^{\texS\times M};\ \text{for each} \ i\in\N_{k_1}.
\end{equation}
The product submatrices can be ordered in terms of the indices $i$ and $j$, e.g., through the bijection $\phi:\N_{k_1}\times\N_{k_2}\to\N_{k}$ defined as $\phi:(i,j)\mapsto(i-1)k_2+j$. Since $k=k_1k_2$, each $\Cbar_{i,j}$ corresponds to one of $k$ distinct subtasks which need to be retrieved.

As mentioned, the main benefit of $\cmmT$ when compared to $\cmmO$, is that the communication load between the $i^{th}$ worker and the central server drops by a factor of $k/|\J_i|$; when considering equipotent partitions of $A$ and of $B$. Therefore, if Algorithm \ref{determine_aI} were to be used for the decoding step, the central server would also require much less temporary storage. The workers on the other hand, need to store the entire matrix $A$.

The idea behind both the decoding Algorithms \ref{determine_aI} and \ref{stream_dec} work, under a slight modification which we explain. In the proposed GCS we dealt with vector addition, and in our first CMMS; with matrix addition. Now, we focus on submatrices of the final product, which is common in CMM \cite{LLPPR18,YMAA17}. Our decoding vector $\ab_{\I}$ will be the same, but the way we apply it is different. If Algorithm \ref{determine_aI} were to be used, every worker corresponding to the same congruence class $c\in\N_{0,s}$ communicates back the same set of computations $\{\Cbar_{i,j}\}_{\phi(i,j)\in\J_{[c]_{s+1}}}$, which sets are distinct for each congruence class. Hence, whenever a complete residue system, in terms of the workers indices, is received, then the central server will have in its possession all the computations $\{\Cbar_{i,j}:i\in\N_{k_1},j\in\N_{k_2}\}$. These computations are then rearranged in order to recover $C$.

If Algorithm \ref{stream_dec} were to be used, the same idea holds. The central server waits until at least one corresponding worker from each block, from one of the two matrices $\Bb_{\Cfr_1}$ or $\Bb_{\Cfr_2}$ has responded. Formally, we want a scheme $(\Bbt,\abt_{[I]_{s+1}})$ such that $\abt_{[I]_{s+1}}^T\Bbt=\Ib_{k_1M}$ (note that $k\texT=k_1k_2\cdot M/k_2=k_1M$) for any $\I\in\I_f^n$, where $[I]_{s+1}$ is the congruence class of the complete residue system present in $\I$. This is analogous to the GC condition $\ab_\I\Bb=\bold{1}_{k\times1}$. To summarize, the encoding process is
\begin{equation}
\label{enc_CMM_2}
  \overbrace{\big(\Ib_{k_1k_2}\otimes\bold{1}_{(s+1)\times 1}\otimes\Ib_{\texT}\big)}^{\Bbt\in\{0,1\}^{k\texT(s+1)\times k\texT}} \cdot \overbrace{\begin{bmatrix} \Cbar_1^T \\ \vdots \\ \Cbar_{k_1}^T \end{bmatrix}}^{\Cbbar^T\in\R^{k\texT\times \texS}} { = \Bbt \cdot \begin{bmatrix} \Cbar_{1,1}^T \\ \Cbar_{1,2}^T \\ \vdots \\ \Cbar_{1,k_2}^T \\ \vdots \\ \Cbar_{k_1,k_2}^T \end{bmatrix} = \begin{bmatrix} \Cbar_{1,1}^T \\ \vdots \\ \Cbar_{1,1}^T \\ \vdots \\ \Cbar_{k_1,k_2}^T \\ \vdots \\ \Cbar_{k_1,k_2}^T \end{bmatrix},}
\end{equation}
where the transpose of each submatrix $\Cbar_{i,j}$ appears $s+1$ times along the rows of the encoding $\Bbt \Cbbar^T\in\R^{(s+1)k_1M\times \texS}$, each corresponding to one of the $s+1$ potential workers that are asked to compute $\Cbar_{i,j}$. This encoding is analogous to the encoding $\Bb\gb$ in our GCS, where the partial gradients $\{g_l\}_{l=1}^k$ correspond to the submatrices $\{\Cbar_{i,j}:i\in\N_{k_1},j\in\N_{k_2}\}$. Furthermore, $\Bbt$ reveals the task allocation which is applied to the pairs $\{(\Abar_i,\Bbar_j):i\in\N_{k_1},j\in\N_{k_2}\}$. Alternatively, the matrix $\Bbt$ is constructed as described in Algorithm \ref{enc_Bt}.

\begin{algorithm}[h]
\label{enc_Bt}  
\SetAlgoLined
  \KwIn{parameters $k_1$, $k_2$, $s$ and $\texT=M/k_2$ \Comment{assume $n=k=k_1k_2$}}
  \KwOut{$\Bbt\in\{0,1\}^{k\texT(s+1)\times k\texT}$}
  \textbf{Initialize:} $\Bb=\bold{0}_{k(s+1)\times k}$\\
  \For{$j=1$ to $k$}
    {$\Bb[(j-1)\cdot(s+1)+1:j\cdot(s+1),j]=\bold{1}_{(s+1)\times1}$}
  \Return $\Bbt\gets\Bb\otimes\Ib_{\texT}$
\caption{Encoding Matrix $\Bbt$ --- $\cmmT$}
\end{algorithm}

For the decoding matrix $\abt_{[I]_{s+1}}$ constructed by Algorithm \ref{determine_aI_matrix}, we use Algorithm \ref{determine_aI} as a subroutine in order to determine which congruence class $I$ of worker indices forms a complete residue system, in the inner \textbf{if} statement. We could directly use the indicator-vector $\text{rec}_{\I}$, though working with $\ab_{\I}$ is preferred. This also reveals how the decoding is similar to that of our GCS. By our assumptions, $\Bbt$ and $\abt_{[I]_{s+1}}$ are both of size $k\texT(s+1)\times k\texT$; where $k\texT=k_1M$. As was previously mentioned, a rearrangement on $\Cbbar^T$ needs to take place to finally recover $C$. In the special case where $k_1=1$, a rearrangement is not necessary.

\begin{algorithm}[h]
\label{determine_aI_matrix}
\SetAlgoLined
  \KwIn{decoding vector $\ab_{\I}$, by invoking Algorithm \ref{determine_aI}, and design parameters $k_1$, $k$, $s$ and $\texT=M/k_1$}
  \KwOut{$\abt_{[I]_{s+1}}\in\{0,1\}^{k\texT(s+1)\times k\texT}$}
  \textbf{Initialize:} $\abt_{\I}=\bold{0}_{k(s+1)\times k}$\\
  \For{$j=1$ to $s+1$}
    {
     \If{$(\ab_{\I})_j \equiv 1$}
       {$I\gets j-1$ }
    }
  \For{$j=1$ to $k$}
    {$\abt_{\I}[(j-1)\cdot(s+1)+(I+1),j]=1$}
  \Return $\abt_{[I]_{s+1}}\gets\abt_{\I}\otimes\Ib_{\texT}$
\caption{Decoding Matrix $\abt_{[I]_{s+1}}$ --- $\cmmT$}
\end{algorithm}

Similar to GC where $\ab_\I$ is constructed such that $\ab_\I^T\Bb=\bold{1}_{1\times k}$, we constructed $\abt_{[I]_{s+1}}$ to satisfy $\abt_{[I]_{s+1}}^T\Bbt=\Ib_{k_1M}$. By this, the above pair $(\abt_{[I]_{s+1}},\Bbt_\I)$ yields
\begin{equation}
  \abt_{[I]_{s+1}}^T(\Bbt_\I \Cbbar^T) = \Ib_{k_1M}\Cbbar^T = \Cbbar^T,
\end{equation}
which implies that our CMM construction works as expected. In appendix \ref{special_case_CMM_app} we provide the analysis for the case where $k_1=1$, as a simpler version of $\cmmT$. A numerical example is depicted in Appendix \ref{example_CMM2}, for parameters $k_1=1$ and $k_2=8$.

\subsection{Comparison Between $\cmmO$ and $\cmmT$}
\label{comp_section}

We present the trade-offs, in terms of communication, storage, and computational operations required by each worker for our CMMSs in Table \ref{CMM_comp_table}. Each scheme may have different uses in practice, in which one is preferable to the other. In certain applications, one may even be easier to implement compared to the other. Also, depending on the limitations and the parameters of the system employing the matrix-matrix multiplications and the underlying application, a different CMMS may be more suitable. Both $\cmmO$ and $\cmmT$ have been applied in other coded computing schemes, in which the other cannot be utilized. The approach of $\cmmO$ was used for approximate matrix-multiplication \cite{CPH20c}, and the special case of $\cmmT$ where $k_1=1$; for matrix inversion \cite{CPH20b}. We briefly explain these applications, after comparing the trade-offs of the two schemes.

\begin{center}
\begin{table}[h]
\centering
\begin{tabular}{ |p{1cm}||p{1.9cm}|p{1.9cm}|p{2.3cm}| }
\hline
\multicolumn{4}{|c|}{Trade-Offs Between Our Two CMM Schemes} \\
\hline
\hline
 & Communication & Computation & Storage \\
\hline
$\cmmO$ & {\footnotesize$LM$} & {\footnotesize$LMN/k\cdot|\J_i|$} & {\footnotesize$\frac{N}{k}(L+M)\cdot|\J_i|$} \\
\hline
$\cmmT$ & {\footnotesize$\frac{LM}{k_1k_2}\cdot|\J_i|$} & {\footnotesize$\frac{LMN}{k_1k_2}\cdot|\J_i|$} & {\footnotesize$N\big(\frac{L}{k_1}+\frac{M}{k_2}\big)\cdot|\J_i|$} \\
\hline
\end{tabular}
\vspace{3mm}
\caption{Comparison of the communication, computation and storage required by the workers in each of our schemes.}
\label{CMM_comp_table}
\end{table}
\end{center}
\vspace{-.8cm}

For a fair comparison between the two schemes, let us assume that the same number of jobs $|\J_i|$ is assigned to every worker across both $\cmmO$ and $\cmmT$. In terms of communication; if the bandwidth is limited, $\cmmT$ is preferred, as each worker only needs to send a fraction of the $LM$ matrix symbols to the central server; since $|\J_i|/(k_1k_2)<1$. In terms of computation, the total number of floating-point operations carried out locally by the workers, is the same in the two schemes, when the total number of subtasks ($k$ and $k_1k_2$ respectively) are equal. The preferred scheme therefore depends on how parameters $k_1$ and $k_2$ are selected for $\cmmT$; relative to $k$ for $\cmmO$, and vice versa. In terms of storage, a similar comparison holds, e.g., if we set $k_1=k_2=k$; the workers in both schemes require the same amount of local storage.

Theorem \ref{thm_CMM} and $\cmmO$ were incorporated in a \textit{weighted CMM} approximation scheme \cite{CPH20c}. The idea behind the weighting is that each outer-product matrix, which is requested to be computed, is scaled by an integer factor corresponding to an importance sampling distribution on the submatrix pairs $\{(\At_j,\Bt_j)\}_{j=1}^k$. The fact that the workers in $\cmmO$ compute the outer-products of column-row submatrix pairs, permits us to combine this approach with $CR$-multiplication; a randomized technique which produces a low-rank approximate product of minimum variance \cite{DK01,DKM06a,DKM06b}, as both $\cmmO$ and $CR$-multiplication leverage \eqref{out_pro_identity}. This procedure resulted in an approximate CMM with reduced storage and number of operations at the workers.

The approach of $\cmmT$ was used in \cite{CPH20b}, as a basis to approximate the inverse and pseudoinverse of a matrix in the presence of stragglers. The analogy which takes place is that instead of the products $A\Bbar_i$, the workers in the matrix inversion scheme communicate back
\begin{equation}
  \Big[\hat{B}^{((s+1)i+1)}\ \cdots \ \hat{B}^{((s+1)(i+1))}\Big]\in\R^{N\times\frac{N}{k}},
\end{equation}
for $i=0,1,\cdots,k-1$ and $k=N/(s+1)$. Matrix $\hat{B}$ is the approximation of the inverse of $A\in\R^{N\times N}$, i.e., $A\hat{B}\approx\Ib_N$, whose columns are estimated by the workers; who approximate the solutions to the regression problems
\begin{equation}
  \hat{B}^{(l)}=\arg\min_{\bb\in\R^N}\left\{ \|A\bb-\eb_l\|_2^2 \right\},
\end{equation}
for each of the columns of $\hat{B}$ requested by them, by using an iterative method of their choice.

\section{Comparison to Prior Works}
\label{comparison_sec}

In this section we briefly review some related work, which we compare our schemes to. We review a polynomial based GC, and three polynomial CMM approaches. We compare and contrast each of the CMM approaches to one of ours ($\cmmO$, $\cmmT$, and the special case of $\cmmT$ presented in Appendix \ref{special_case_CMM_app}), which are in fact the closest line of work we could find to each of our proposed schemes. Generally speaking, the communication, storage and computation required by our CMMSs is the same with the respective one we compare it to.

The advantage of the CMM polynomial schemes we will discuss is in terms of the recovery threshold. These schemes achieve a better threshold, as they encode linear combinations of submatrices of one or both the matrices, and then carry out the computation on the encoded submatrices, from which a fraction of all the assigned tasks they then decode to retrieve the matrix product. As in GC, in our CMMSs we first carry out the computations and then encode them locally at the worker nodes, e.g., \eqref{enc_dec_CMM} and \eqref{enc_CMM_2}. Our CMMSs meet the optimal recovery threshold known for GC, as this is met by the underlying GCS; which we proposed. However, our schemes are superior in terms of encoding and decoding complexity. Furthermore, since the encodings are binary matrices consisting only of $0$'s and $1$'s, they introduce no numerical inaccuracies nor rounding errors.

We also draw connections with weighted GC, distributed storage systems, and $\ldpc$ codes.

\subsection{Reed-Solomon Scheme and Weighted Gradient Coding}

First, we compare our proposed GCS with the one introduced in \cite{HASH17}, which provides improvements in terms of the decoding complexity to \cite{TLDK17} and, to the best of authors' knowledge, is the first work to consider constructing the decoding vector $\ab_\I$ online; instead of the matrix comprised of all possible $\ab_{\I}$ decoding vectors.

The main idea in \cite{HASH17} is to use balanced Reed-Solomon codes \cite{HLH16}, which are evaluation polynomial codes. Each column of the encoding matrix $\Bb$ corresponds to a partition $\D_i$ and is associated with a polynomial that evaluates to zero at the respective workers who have not been assigned that partition part. The properties of balanced Reed-Solomon codes imply the decomposition $\Bb_{\I}=\Gb_{\I}\Tb$, where $\Gb_{\I}$ is a Vandermonde matrix over certain roots of unity, and the entries of $\Tb$ corresponds to the coefficients of polynomials; constructed such that their constant term is $1$, i.e., $\Tb_{(1)}=\bold{1}_{1\times k}$. The decoding vector $\ab_{\I}^T$ is the first row of $\Gb_{\I}^{-1}$, for which $\ab_{\I}^T\Gb_{\I}=\bold{e}_1^T$. A direct consequence of this is that $\ab_{\I}^T\Bb_{\I}=\bold{e}_1^T\Tb=\Tb_{(1)}$, thus $\ab_{\I}^T(\Bb_{\I}\gb)=g^T$.

A drawback of this construction is that it works over the complex numbers and requires an inversion, which introduces numerical inaccuracies, due to the fact that the encoding and decoding computations are done with finite precision. Each decoding vector $\ab_{\I}$ can be computed in time $O((n-s)^2)$, while the decoding vector proposed in this paper is constructed in time $O(n+s)$.

The Reed-Solomon based scheme was also used as a basis for \textit{weighted gradient coding} \cite{CPH20a}. The idea behind the weighting is similar to that of the weighted CMMS \cite{CPH20c} described in Subsection \ref{comp_section}. In \textit{weighted GC} the goal is not to recover the sum of the partial gradients, but a weighted linear combination according to some predetermined weight vector $\wb\in\Z_{+}^{1\times k}$. This has many applications in signal processing and statistics. Note also that our proposed binary gradient code $(\Bb,\ab_\I)$ can be extended to a weighted scheme $(\Bbt,\ab_\I)$ by the simple modification of $\Bb$:
\begin{equation}
  \Bbt_{ij} = \begin{cases} \wb_{j} &\mbox{if } \ \Bb_{ij}=1 \\ 0 &\mbox{if } \ \Bb_{ij}=0 \end{cases} \ \ .
\end{equation}

\subsection{CMM \sfsty{MatDot} Codes}
\label{CMM_MatDot}

The proposed $\cmmO$ described in Subsection \ref{1st_sch_sec} is in nature, close to the ``\sfsty{MatDot} Codes'' from \cite{FJHDCG17,DFHJCG19}, which work with the rank-$\tau$ outer-products. In the \sfsty{MatDot} procedure, a polynomial of the submatrices $\At_i$ and $\Bt_i$ is evaluated, i.e.,
\begin{equation}
  p_{A}(x)=\sum_{j=1}^{k}\At_{j}x^{j-1} \quad \text{ and } \quad p_{B}(x)=\sum_{j=1}^{k}\Bt_{j}x^{k-j},
\end{equation}
over arbitrary distinct elements $x_1,\cdots,x_n$ of a finite field $\F_q$ for some $q>n$. The $\iota^{th}$ worker receives the evaluations $p_A(x_\iota)$ and $p_B(x_\iota)$, i.e., the evaluations of the polynomials at the evaluation point corresponding to the worker; $x_\iota$. Each worker is the requested to communicate back the computation $P(x_i)=p_A(x_i)p_B(x_i)$, which is a polynomial of degree $2(k-2)$. The sum of all the outer-products is the coefficient of $x^{k-1}$ of the polynomial $p_A(x)p_B(x)$. Once any $2k-1$ evaluations of the polynomial $P(x)$ on distinct points are received, polynomial interpolation of Reed-Solomon decoding can be applied in order to retrieve the product $AB$ \cite{FJHDCG17,DFHJCG19}.

The \sfsty{MatDot} procedure in \cite{FJHDCG17,DFHJCG19} is described for $A$ and $B$ both being square matrices of size $N\times N$, though there is no reason why it should not work for non-square matrices. The overall encoding complexity for each worker if both matrices considered are squares; is $O(N^2n)$. The overall decoding complexity per worker is $O(k\log^2k)$ for each entry \cite{Kun73}, thus; the overall decoding complexity is $O(N^2k\log^2k)$.

The communication cost per worker of the \sfsty{MatDot} scheme is the same as $\cmmO$. A drawback is the storage required at the central server, which was overcome for our scheme through the alternative decoding process of Algorithm \ref{stream_dec}.

\subsection{CMM Polynomial Codes}
\label{CMM_poly_codes}

The ``Polynomial Codes'' proposed in \cite{YMAA17} follow a similar approach in terms of computational tasks and concatenation as the proposed $\cmmT$, though a different encoding and decoding procedure is considered. The Polynomial Codes CMMS partitions both the matrices, $A$ into $k_2$ submatrices across its rows; and $B$ into $k_2$ submatrices across its columns. The encodings which take place are similar to those of \sfsty{MatDot} Codes, and the workers are requested to compute the product between an encoding of the submatrices $A$ and of the submatrices $B$. Once $k_1k_2$ workers respond, the decoding involves the inversion of a Vandermonde matrix; which introduces numerical inaccuracies, in order to retrieve the submatrices of $C$ each of size $\frac{L}{k_1}\times\frac{M}{k_2}$, which are then concatenated.

A restriction of the Polynomial Codes, which cannot be altered if we want to have a recovery threshold of $k_1k_2$, is the fact that the resulting products of the encoded submatrices all need to have the same size, therefore requiring that $k_1|L$ and $k_2|M$. From the analysis we carried for homogeneous workers of our GCS, the partitions of $A$ and $B$ for our $\cmmT$ do not all need to have the same number of columns.

An extension of the Polynomial codes in \cite{YMAA20} does the same augmentation argument after the decoding step, though their encodings take place over submatrices of $A$ and $B$, where the partitions are carried out for both matrices across the rows and columns.

\subsection{Connection to Distributed Storage Systems}
\label{distr_stor_subsec}

A central theme of this paper was relaxing the condition that $(s+1)\mid n$, in order to design a GCS for a pair of integers $(s,n)$ where $0\leq s<n$ for which $(s+1)\nmid n$. The main idea behind this condition is that the $k$ partitions can be appropriately grouped together and the workers will all get the same number of partitions. This is what is referred to as uniform, defined in Definition \ref{def_unif}. If $d_s(\Bb)=0$, the assignment is balanced, according to the terminology in \cite{TLDK17}. The arithmetic is easier to work with under this assumption, which is also why our construction results in a block diagonal matrix $\Bb$; when $(s+1)\mid n$. This is a permutation of the rows of the fractional repetition scheme from \cite{TLDK17}.

The aforementioned assumption prevails in the construction of \textit{distributed storage systems} as well, which use similar techniques, including replication and block coding to provide reliability. Specifically, in the design of \textit{locally repairable codes} (LRCs) \cite{PD14}. We relate these two applications; of GC and distributed storage, by indicating how this assumption translates from LRCs to GC schemes, by discussing an analog of \cite[Remark 2]{PD14} in this context. The reader is referred to \cite{PD14} and the references therein for further details on LRCs. 

\begin{Rmk}
Observe that when $(s+1)\mid n$, we can partition the set of $k$ disjoint parts of $\D$ into $\frac{n}{s+1}$ disjoint $(s+1)$-groups. The method used in \cite{TLDK17} to prove the lower bound $\|\Bb_{(i)}\|_0\geq \frac{k}{n}(s+1)$ of \cite[Theorem 1]{TLDK17}, relies on the fact that at least one encoding of all of the $(s+1)$-groups need be collected, in order for the scheme to be resilient to any $s$ stragglers. The construction of their code meets this bound with equality, as does ours; under the given assumption. That is, the constructions gives an achievability proof for the case of $(s+1)\mid n$. Furthermore, we show that the pair-wise disjoint parts $\D$ is one of the possibly many arrangements of repair groups that leads to optimal constructions (Theorem \ref{thm_optimality}), as we can rearrange any of the allocation of the parts among each congruence class of workers. This is done in such a way that each partition is allocated to exactly one worker per class.
\end{Rmk}

\subsection{Connection to $\ldpc$ Codes}
\label{LDPC_subsec}

In terms of error-correcting codes, Theorem \ref{cond_B_thm} suggests an analogy between permutations of the task allocations per congruence class $\bar{\Bb}$ of our encoding matrix $\Bb$, and parity check matrices of $\ldpc$ codes. Specifically, $\bar{\Bb}$ matches the definition of an irregular $\ldpc$ parity check matrix $\bold{H}$ \cite{Gal62,mac96,RSU01}, since along the columns of $\bar{\Bb}$ and $\Bb$ we have a balanced load, and along the rows we have an unbalanced load when $(s+1)\nmid n$. That is, $\bar{\Bb}$ suggested by our construction and Theorem \ref{cond_B_thm} is a valid $\bold{H}$, with the additional constraint that any two rows corresponding to the same congruence class $\bmod(s+1)$ have disjoint supports. It is intriguing to see what can be inferred from the constructions and theory of our binary encoding to that of $\ldpc$ codes and vice versa, and if further connections can be established.

\section*{Acknowledgments}
\noindent The research presented in this paper was partially supported by the U.S. National Science Foundation (NSF) under grants NSF CCF-2246213 and CCF-2312752, the U.S. Army Research Office under grant W911NF2310343, the U.S. Department of Energy under grant DE-NA0003921, the Center for Ubiquitous Connectivity (CUbiC) under the JUMP 2.0 program, and the EnCORE Institute under the NSF grant 2217058. We thank the anonymous reviewers and Editors for their constructive comments, for carefully proofreading our work, and providing constructive comments and feedback. We are also grateful to Mert Pilanci for sharing with us the completion times of the experiments carried out in \cite{BP23}.

\section{Conclusion and Future Work}
\label{sec:conc}

In this paper, we introduced a binary GCS for distributed optimization. The main advantages of our code design is that (i) it provides numerically accurate and efficient computations of the gradient, (ii) it removes the limiting assumption $(s+1)\mid n$, and (iii) it has an improved and more tractable decoding stage compared to that of the first GCS proposed in \cite{TLDK17}. We provided an analysis of the proposed design, and showed that it is optimal in terms of the minimizing function $d_s$ defined in Definition \ref{def_unif}. Both homogeneous and heterogeneous workers were considered. It is worth noting that more recent work also considers this direction \cite{SR23}, though their focus is on improving the recovery threshold. We then presented two CMM approaches; as extensions of our binary GCS.

There are several interesting directions for future work. We have seen that the proposed schemes accommodate various matrix operations. It would be interesting to see what other operations they can accommodate, in order to devise exact and approximate straggler resilient coded computing schemes. Another direction is to incorporate privacy and security into our schemes. A third direction, is to further explore the connections between coded computations and codes for distributed storage systems. Specifically, it would be worthwhile to explore the connections between the proposed GCS, the GCS of \cite{KKR19}, and the distributed storage systems of \cite{ERK10}, which we briefly described in Subsection \ref{distr_stor_subsec}.

\appendices
\section{Pseudocode of Encoding Matrices $\tilde{\Bb}_{\Cfr_1}$ and $\tilde{\Bb}_{\Cfr_2}$}
\label{appendix_Enc_matr}

In this appendix we provide the pseudocode of the encoding matrices $\tilde{\Bb}_{\Cfr_1}$ and $\tilde{\Bb}_{\Cfr_2}$ described in Subsections \ref{subsec_B_C1} and \ref{subsec_B_C2}, which are combined to give the encoding matrix $\Bb$ of our GCS.

\begin{algorithm}[h]
\label{B_ones_unb_unif_C1}
\SetAlgoLined
{\footnotesize
  \KwIn{number of workers $n$ and stragglers $s$, where $s,n\in\Z_+$}
  \KwOut{encoding matrix $\tilde{\Bb}_{\Cfr_1}\in\{0,1\}^{n\times n}$ 
  \Comment{assume $n=k$}}
  $\tilde{\Bb}_{\Cfr_1}\gets\bold{0}_{n\times n}$, and use the division algorithm to get the parameters: \\
  $\ind n=\ell\cdot(s+1)+r \qquad r=t\cdot\ell+q \qquad n=\lambda\cdot(\ell+1)+\rt$ \\

  \For{$i=0$ to $r-1$}
  {
    \If{$\ell+r>s$}
    {
      \For{$j=1$ to $\ell+r-s$}
      {
        $\tilde{\Bb}_{\Cfr_1}\Big[(j-1)(s+1)+i,(j-1)(s+1)+1:j(s+1)\Big]=\bold{1}_{1\times(s+1)}$ \\
      }
      \For{$j=\ell+r-s+1$ to $\ell+1$}
      {
        $\tilde{\Bb}_{\Cfr_1}\Big[(j-1)(s+1)+i,(j-1)s+(\ell+r-s)+1:(j-1)s+\ell+r\Big]=\bold{1}_{1\times s}$
      }
    }
    \ElseIf{$\ell+r\leq s$}
    {
      \For{$j=1$ to $\rt$}
      {
        $\tilde{\Bb}_{\Cfr_1}\Big[(j-1)(s+1)+i,(j-1)(\lambda+1)+1:j(\lambda+1)\Big]=\bold{1}_{1\times(\lambda+1)}$ \\
      }
      \For{$j=\rt+1$ to $\ell+1$}
      {
        $\tilde{\Bb}_{\Cfr_1}\Big[(j-1)(s+1)+i,(j-1)\lambda+\rt+1:(j-1)\lambda+\rt+\lambda\Big]=\bold{1}_{1\times\lambda}$
      }
    }
  }
 \Return $\tilde{\Bb}_{\Cfr_1}$
 \caption{Encoding $\tilde{\Bb}_{\Cfr_1}$ --- $\Cfr_1=\left\{[i]_{s+1}\right\}_{i=0}^{r-1}$}
}
\end{algorithm}

\begin{algorithm}[h]
\label{B_ones_unb_unif_C2}
\SetAlgoLined
{\footnotesize
  \KwIn{number of workers $n$ and stragglers $s$, where $s,n\in\Z_+$}
  \KwOut{encoding matrix $\tilde{\Bb}_{\Cfr_2}\in\{0,1\}^{n\times n}$ 
  \Comment{assume $n=k$}}
  $\tilde{\Bb}_{\Cfr_2}\gets\bold{0}_{n\times n}$, and use the division algorithm to get the parameters: \\
  $\ind n=\ell\cdot(s+1)+r \qquad r=t\cdot\ell+q \qquad n=\lambda\cdot(\ell+1)+\rt$ \\

  \For{$i=r$ to $s$}
  {
    \If{$q\equiv0$}
    {
      \For{$j=1$ to $\ell$}
      {
        $\tilde{\Bb}_{\Cfr_2}\Big[(j-1)(s+1)+i,(j-1)(s+t+1)+1:j(s+t+1)\Big]=\bold{1}_{1\times(s+t+1)}$ \\
      }
    }
    \ElseIf{$q>0$}
    {
      \For{$j=1$ to $q$}
      {
        $\tilde{\Bb}_{\Cfr_2}\Big[(j-1)(s+1)+i,(j-1)(s+t+2)+1:j(s+t+1)\Big]=\bold{1}_{1\times(s+t+2)}$ \\
      }
      \For{$j=q+1$ to $\ell$}
      {
        $\Bb\Big[(j-1)(s+1)+i,(j-1)(s+t+1)+q+1:j(s+t+1)+q\Big]=\bold{1}_{1\times(s+t+1)}$
      }
    }
  }
 \Return $\tilde{\Bb}_{\Cfr_2}$
 \caption{Encoding $\tilde{\Bb}_{\Cfr_2}$ --- $\Cfr_2=\left\{[i]_{s+1}\right\}_{i=r}^{s}$}
}
\end{algorithm}

\section{Special Case of $\cmmT$}
\label{special_case_CMM_app}

In this appendix, we present a special of $\cmmT$; where $k_1=1$ and $k=k_2$, as this may have certain applications where one cannot partition both matrices, e.g., \cite{CPH20c}. Additionally, it is simpler to understand this case, and then view $\cmmT$ as applying this special case of the code $k_1$ times.

In contrast to the partitioning \eqref{block_partition} of $\cmmO$; and the general case of $\cmmT$ \eqref{part_ccm2}, in this approach we partition only one of the two matrices, say $B$; along its columns. Each worker computes the product of a submatrix of $B$ with the matrix $A$, and then the central server augments the received computations accordingly. The key property that we utilize is the following:
\begin{equation}
  AB = A\cdot\Big[\Bbar_1 \ \cdots \ \Bbar_k\Big] = \Big[A\Bbar_1\ \cdots \ A\Bbar_k\Big] = \Big[\Ct_1 \ \cdots \ \Ct_k\Big]
\end{equation}
where $\Ct_j\coloneqq A\Bbar_j$, for all $j=1,\cdots,k$. For convenience we assume that $(s+1)\mid M$, and similar to our GCS that $n=k$. To further simplify our construction and description, under the assumption that $(s+1)\mid M$, we can assume that the equipotent partitions of $\{\Bbar_i\}_{i=1}^k$ which are distributed to the workers are of size $\texT=M/k\in\Z_+$, which implies that $(s+1)\mid k$. All in all, each worker receives only one $\Bbar_j\in\R^{N\times T}$; where $\texT$ is a multiple of $(s+1)$, and $\Ct_j\in\R^{L\times T}$. We also note that a similar multiplication can take place if we instead partition $A$ along its rows --- this corresponds to the case where $k_2=1$ and $k=k_1$. For completeness, we present the corresponding encoding through matrix $\Bbt\in\{0,1\}^{k\texT(s+1)\times k\texT}$
\begin{equation}
\label{enc_CMM_2_spec_case}
  \overbrace{\big(\Ib_k\otimes\bold{1}_{(s+1)\times 1}\otimes\Ib_{\texT}\big)}^{\Bbt\in\{0,1\}^{k\texT(s+1)\times k\texT}} \cdot \begin{bmatrix} \Ct_1^T \\ \vdots \\ \Ct_k^T \end{bmatrix} = \begin{bmatrix} \Ct_1^T \\ \vdots \\ \Ct_1^T \\ \vdots \\ \Ct_k^T \\ \vdots \\ \Ct_k^T \end{bmatrix},
\end{equation}
where the transpose of each submatrix $\Ct_j$ appears $s+1$ times along the rows of the encoding $\Bbt C^T\in\R^{(s+1)M\times L}$, each corresponding to one of the $s+1$ potential workers that are asked to compute $\Ct_j$. Lastly, in this special case of $\cmmT$, no rearrangement is needed once the decoding step has been applied.

\section{Numerical Example of the Proposed Encodings and Decodings}
\label{example_app}

In this appendix, we give examples of our encoding and decoding algorithms, to convey the main ideas and help visualize the task allocations which take place. For our GCS of Section \ref{BGC_sec}, consider the case where $n=k=11$ and $s=3$.

By \eqref{eq_1}, \eqref{eq_2}, \eqref{eq_3} we then have $\ell=2,r=3,t=1,q=1$; thus $\ell>r-s$, and the task allocation for $\Cfr_1$ is described by $\Bb_{\Cfr_1}\in\{0,1\}^{(\ell+1)r\times n}$:
$$ \Bb_{\Cfr_1} = \begin{bmatrix} 
\textbf{\blue 1} & \textbf{\blue 1} & \textbf{\blue 1} & \textbf{\blue 1} & & & & & & & \\
\textit{\cyan 1} & \textit{\cyan 1} & \textit{\cyan 1} & \textit{\cyan 1} & & & & & & & \\
\mathfrak{\purple 1} & \mathfrak{\purple 1} & \mathfrak{\purple 1} & \mathfrak{\purple 1} & & & & & & & \\
& & & & \textbf{\blue 1} & \textbf{\blue 1} & \textbf{\blue 1} & \textbf{\blue 1} & & & \\
& & & & \textit{\cyan 1} & \textit{\cyan 1} & \textit{\cyan 1} & \textit{\cyan 1} & & & \\
& & & & \mathfrak{\purple 1} & \mathfrak{\purple 1} & \mathfrak{\purple 1} & \mathfrak{\purple 1} & & & \\
& & & & & & & & \textbf{\blue 1} & \textbf{\blue 1} & \textbf{\blue 1} \\
& & & & & & & & \textit{\cyan 1} & \textit{\cyan 1} & \textit{\cyan 1} \\
& & & & & & & & \mathfrak{\purple 1} & \mathfrak{\purple 1} & \mathfrak{\purple 1}
\end{bmatrix}, $$
where each congruence class is represented by a different color and font for clarity. The zero entries are omitted. The indicated dimensions are for the case where $r>0$, i.e., the remainder block is not empty. The encoding corresponding to the congruence classes $0$ to $r-1$ constructed by Algorithm \ref{B_ones_unb_unif_C1}, is obtained from $\Bb_{\Cfr_1}$ by properly appending zero vectors. Specifically, $\Bb_{\Cfr_1}$ is the restriction of $\tilde{\Bb}_{\Cfr_1}$ to the rows with nonzero entries.

For the remaining congruence classes, $r$ to $s$, since $q=1$; we have:
$$ \Bb_{\Cfr_2} = \begin{bmatrix} 
{\dgreen 1} & {\dgreen 1} & {\dgreen 1} & {\dgreen 1} & {\dgreen 1} & {\dgreen 1} & & & & & \\
& & & & & & {\dgreen 1} & {\dgreen 1} & {\dgreen 1} & {\dgreen 1} & {\dgreen 1}
\end{bmatrix}, $$
where $\Bb_{\Cfr_2}\in\{0,1\}^{\ell(s+1-r)\times n}$ is the restriction of $\tilde{\Bb}_{\Cfr_2}$ constructed by Algorithm \ref{B_ones_unb_unif_C2} to the rows with nonzero entries.

The final step is to appropriately merge the two matrices together, so that the congruence classes are in ascending order. Considering Algorithms \ref{B_ones_unb_unif_C1} and \ref{B_ones_unb_unif_C2}, this corresponds to $\Bb=\tilde{\Bb}_{\Cfr_1}+\tilde{\Bb}_{\Cfr_2}$. We therefore get the following encoding matrix:
$$ \Bb = \begin{bmatrix} 
\textbf{\blue 1} & \textbf{\blue 1} & \textbf{\blue 1} & \textbf{\blue 1} & & & & & & & \\
\textit{\cyan 1} & \textit{\cyan 1} & \textit{\cyan 1} & \textit{\cyan 1} & & & & & & & \\
\mathfrak{\purple 1} & \mathfrak{\purple 1} & \mathfrak{\purple 1} & \mathfrak{\purple 1} & & & & & & & \\
{\dgreen 1} & {\dgreen 1} & {\dgreen 1} & {\dgreen 1} & {\dgreen 1} & {\dgreen 1} & & & & & \\
& & & & \textbf{\blue 1} & \textbf{\blue 1} & \textbf{\blue 1} & \textbf{\blue 1} & & & \\
& & & & \textit{\cyan 1} & \textit{\cyan 1} & \textit{\cyan 1} & \textit{\cyan 1} & & & \\
& & & & \mathfrak{\purple 1} & \mathfrak{\purple 1} & \mathfrak{\purple 1} & \mathfrak{\purple 1} & & & \\
& & & & & & {\dgreen 1} & {\dgreen 1} & {\dgreen 1} & {\dgreen 1} & {\dgreen 1} \\
& & & & & & & & \textbf{\blue 1} & \textbf{\blue 1} & \textbf{\blue 1} \\
& & & & & & & & \textit{\cyan 1} & \textit{\cyan 1} & \textit{\cyan 1} \\
& & & & & & & & \mathfrak{\purple 1} & \mathfrak{\purple 1} & \mathfrak{\purple 1}
\end{bmatrix} \in \{0,1\}^{n\times n}. $$

As discussed in Subsection \ref{val_opt_subsec}, one can apply permutations to the rows and columns of $\Bb$ and obtain a solution to $\IPB$. Furthermore, one can apply a different permutation on the columns for each congruence class, and still obtain a valid GCS. An example of how the allocations can be modified for each congruence class is given below, where the superposition of any set of row vectors of the same color result in $\bold{1}_{1\times11}$, and all colors appear exactly once in each column:

$$ \bar{\Bb} = \begin{bmatrix} 
\textbf{\blue 1} & \textbf{\blue 1} & \textbf{\blue 1} & \textbf{\blue 1} & & & & & & & \\
& \textit{\cyan 1} & & \textit{\cyan 1} & & \textit{\cyan 1} & & \textit{\cyan 1} & & & \\
& & \mathfrak{\purple 1} & & \mathfrak{\purple 1} & & \mathfrak{\purple 1} & & \mathfrak{\purple 1} & & \\
{\dgreen 1} & {\dgreen 1} & & {\dgreen 1} & & & {\dgreen 1} & & {\dgreen 1} & & {\dgreen 1} \\
& & & & \textbf{\blue 1} & \textbf{\blue 1} & \textbf{\blue 1} & \textbf{\blue 1} & & & \\
& & \textit{\cyan 1} & & \textit{\cyan 1} & & \textit{\cyan 1} & & \textit{\cyan 1} & & \\
\mathfrak{\purple 1} & \mathfrak{\purple 1} & & & & & & & & \mathfrak{\purple 1} & \mathfrak{\purple 1} \\
& & {\dgreen 1} & & {\dgreen 1} & {\dgreen 1} & & {\dgreen 1} & & {\dgreen 1} & \\
& & & & & & & & \textbf{\blue 1} & \textbf{\blue 1} & \textbf{\blue 1} \\
\textit{\cyan 1} & & & & & & & & & \textit{\cyan 1} & \textit{\cyan 1} \\
& & & \mathfrak{\purple 1} & & \mathfrak{\purple 1} & & \mathfrak{\purple 1} & & & \end{bmatrix} . $$

\subsection{Example of Algorithm \ref{stream_dec}}
\label{example_stream_dec}

We now demonstrate the decoding procedure of Algorithm \ref{stream_dec}. Considering our earlier example, we have the corresponding encoding matrices $\Bb_{\Cfr_1}$ and $\Bb_{\Cfr_2}$:
$$ \Bb_{\Cfr_1} = \begin{bmatrix} 
\textbf{\blue 1} & \textbf{\blue 1} & \textbf{\blue 1} & \textbf{\blue 1} & & & & & & & \\
\textbf{\blue 1} & \textbf{\blue 1} & \textbf{\blue 1} & \textbf{\blue 1} & & & & & & & \\
\textbf{\blue 1} & \textbf{\blue 1} & \textbf{\blue 1} & \textbf{\blue 1} & & & & & & & \\
& & & & \textit{\cyan 1} & \textit{\cyan 1} & \textit{\cyan 1} & \textit{\cyan 1} & & & \\
& & & & \textit{\cyan 1} & \textit{\cyan 1} & \textit{\cyan 1} & \textit{\cyan 1} & & & \\
& & & & \textit{\cyan 1} & \textit{\cyan 1} & \textit{\cyan 1} & \textit{\cyan 1} & & & \\
& & & & & & & & \mathfrak{\purple 1} & \mathfrak{\purple 1} & \mathfrak{\purple 1} \\
& & & & & & & & \mathfrak{\purple 1} & \mathfrak{\purple 1} & \mathfrak{\purple 1} \\
& & & & & & & & \mathfrak{\purple 1} & \mathfrak{\purple 1} & \mathfrak{\purple 1}
\end{bmatrix} $$
and
$$ \Bb_{\Cfr_2} = \begin{bmatrix} 
{\dgreen 1} & {\dgreen 1} & {\dgreen 1} & {\dgreen 1} & {\dgreen 1} & {\dgreen 1} & & & & & \\
& & & & & & {\diffred \textbf{I}} & {\diffred \textbf{I}} & {\diffred \textbf{I}} & {\diffred \textbf{I}} & {\diffred \textbf{I}}
\end{bmatrix}. $$
Here, each set of computations is represented by a different color and font, and we require that at least one worker of each of color from one of the two encoding matrices has responded.

In Algorithm \ref{stream_dec} we add only the first received computation of each represented color and font to $C$. In contrast to Algorithm \ref{determine_aI} where we required all workers of a single color to respond in order to perform the decoding step, now we require at least one worker from each block, from either $\Bb_{\Cfr_1}$ or $\Bb_{\Cfr_2}$ to respond. We do this separately for the subroutines corresponding to $\Bb_{\Cfr_1}$ and $\Bb_{\Cfr_2}$, and terminate whenever at least one worker from each block, from one of the two groups $\Cfr_1$ or $\Cfr_2$ has responded. This is equivalent to decoding a repetition erasure code. Lastly, recall that the number of workers per block in $\Bb_{\Cfr_2}$ is $s+1-r$, which in our toy example happens to be one.

\subsection{Example of $\cmmT$, with $k_1=1$}
\label{example_CMM2}

An example of $\cmmT$ with $k_1=1$ is provided to help visualize the encoding task assignments, as well as the decoding. Let $n=k_2=8$, $s=1$ and $M$ be arbitrary, with $\texT=M/k_2$ and $\texS=L/1$. Let $\bold{0}_{\texT}$ denote the $\texT\times\texT$ zero matrix. For $I=1$ and $[I]_{s+1}=\{1,3,5,7\}$, the encoding-decoding pair is
$$ \Bbt = \begin{bmatrix} {\cyan \Ib_{\texT}} & & & \\ {\purple \mathbb{I}_{\mathbb{T}}} & & & \\ & {\cyan \Ib_{\texT}} & &  \\ & {\purple \mathbb{I}_{\mathbb{T}}} & & \\ & & {\cyan \Ib_{\texT}} & \\ & & {\purple \mathbb{I}_{\mathbb{T}}} & \\ & & & {\cyan \Ib_{\texT}} \\ & & & {\purple \mathbb{I}_{\mathbb{T}}} \end{bmatrix} \ , \ \abt_{[I]_{s+1}} = \begin{bmatrix} {\cyan \bold{0}_{\texT}} & & & \\ {\purple \mathbb{I}_{\mathbb{T}}} & & & \\ & {\cyan \bold{0}_{\texT}} & &  \\ & {\purple \mathbb{I}_{\mathbb{T}}} & & \\ & & {\cyan \bold{0}_{\texT}} & \\ & & {\purple \mathbb{I}_{\mathbb{T}}} & \\ & & & {\cyan \bold{0}_{\texT}} \\ & & & {\purple \mathbb{I}_{\mathbb{T}}} \end{bmatrix}
$$
where both matrices are of the same size. From this example, it is also clear that $\abt_{[I]}$ is in fact the restriction of $\Bbt$ to the workers corresponding to the congruence class $I$.

\section{Proofs of Section \ref{proposed_scheme_section}}
\label{app_sect_4}

In this appendix we present proofs of Theorems \ref{thm_optimality} and \ref{cond_B_thm}.

\begin{proof}{[Theorem \ref{thm_optimality}]}
To simplify the presentation of the proof, we restrict it to the case where $n=k$. As was previously mentioned, to meet this assumption we can incorporate instances of the data point $(\bold{0}_{p\times 1},0)$ until $N'$ points total are considered; such that $n\mid N'$, and let $k=n$.

Firstly, in the simplest case where $(s+1)\mid n$; we have $\|\Bb_{(i)}\|_0=\frac{k}{n}(s+1)=\frac{k}{\ell}=s+1$ for all $i\in\N_{0,n}$. This results in a block diagonal matrix $\Bb=\bold{1}_{(s+1)\times (s+1)}\otimes \Ib_\ell\in\{0,1\}^{n\times k}$, which is a permutation of the FRC scheme presented in \cite{TLDK17}, for which $d_s(\Bb)=0$; and $|\supp(\Bb)|=\ell\cdot(s+1)^2=k\cdot(s+1)$. The minimum value of the function $d_s$ is therefore attained, while meeting all inequality constraints with an upper bound of $0$. By Proposition \ref{general_cond_prop} and the fact that $\Bb$ is block diagonal with $\Bb_{ij}=1$ if $\Bb_{ij}\neq0$, the fourth constraint is also met.

For the remainder of the proof, we will consider the case where $(s+1)\nmid n$. We first show the reduction from the optimization problem $\OPR$ to the binary integer program $\IPB$ in which we are only considering \textit{binary} GC schemes, and show that our construction through Algorithms \ref{B_ones_unb_unif_C1} and \ref{B_ones_unb_unif_C2} meets the constraints of $\IPB$. We then argue by contradiction, to show that our construction is a solution to $\IPB$.

By imposing the constraints that $\Bb\in\{0,1\}^{n\times k}$ and $\ab_\I\in\{0,1\}^{n}$ for each $\I\in\I_f^n$, the fourth constraint of $\OPR$ can be replaced by \eqref{cond_gen_GC_B} for all $i\in\N_{0,s}$, and by the fact that
\begin{align*}
  \nnz(\Bb) &= \sum_{\iota=1}^{s+1}\sum_{j\in\K_\iota}\|\Bb_{(j)}\|_0\\
  &= \sum_{\iota=1}^{s+1}\|\bold{1}_{1\times k}\|_0\\
  &= \sum_{\iota=1}^{s+1}k\\
  &= k\cdot(s+1)
\end{align*}
we can drop the first constraint of $\OPR$. We therefore have the following binary integer formulation $\IPB$ of $\OPR$, where we require the additional constraints that the encoding matrix $\Bb$ and decoding vectors $\ab_\I$ are over $\{0,1\}$:
\begin{equation*}
\begin{aligned}
\IPB \qquad \arg\min_{\Bb\in\{0,1\}^{n\times k}} \quad & \big\{d_s(\Bb)\big\}\\
\textrm{s.t.} \ \ & \bigsqcup_{i=0}^{s}\K_i=\N_{0,n-1} \ : \ \big||\K_j|-|\K_l|\big|\leq1, \ \forall j,l\in\N_{0,s} \\
& \big|\|\Bb_{(j)}\|_0-\|\Bb_{(l)}\|_0\big|\leq 1, \ \forall j,l\in \K_i, \ \forall i\in\N_{0,s}\\
& \sum_{j\in\K_i}\Bb_{(j)}=\bold{1}_{1\times k}, \ \forall i\in\N_{0,s}
\end{aligned}\ .
\end{equation*}
For our construction through Algorithms \ref{B_ones_unb_unif_C1} and \ref{B_ones_unb_unif_C2}, the partitioning is in terms of congruence classes is $\left\{[i]_{s+1}\right\}_{i=0}^s=\Cfr_1\bigsqcup\Cfr_2=\N_{0,n-1}$, for which $\big||[j]_{s+1}|-|[l]_{s+1}|\big|=1$ if $[j]_{s+1}$ and $[l]_{s+1}$ are in distinct $\Cfr_\iota$'s, and $\big||[j]_{s+1}|-|[l]_{s+1}|\big|=0$ if they are in the same. Hence, the first constraint of $\IPB$ is met.

Considering our partitioning of $\N_{0,n-1}$, the second constraint of $\IPB$ can be split into the cases:
\begin{enumerate}[label=(\alph*)]
  \item $\big|\|\Bb_{(j)}\|_0-\|\Bb_{(l)}\|_0\big|\leq 1$, for all $j,l\in \Cfr_1$
  \item $\big|\|\Bb_{(j)}\|_0-\|\Bb_{(l)}\|_0\big|\leq 1$, for all $j,l\in \Cfr_2$
\end{enumerate}
where (a) and (b) correspond to Algorithms \ref{B_ones_unb_unif_C1} and \ref{B_ones_unb_unif_C2} respectively. By construction, it is clear that both (a) and (b) are met. The final constraint of $\IPB$ for the partitioning through congruence classes, can be reformulated as in \eqref{sum_cond_BGC}, which is also met by construction; for each $c\in\N_{0,s}$. 

Now, for a contradiction, assume that there is a $\Bb'\in\{0,1\}^{n\times k}$ that satisfies the three constraints of $\IPB$, for which $d_s(\Bb')<d_s(\Bb)$. Denote the respective summands by $d_i'\coloneqq\big|\|\Bb_{(i)}'\|_0-(s+1)\big|$ and $d_i\coloneqq\big|\|\Bb_{(i)}\|_0-(s+1)\big|$; for each $i\in\N_{0,n-1}$. That is, $d(\Bb')=\sum_{j=0}^{n-1} d_j'$ and $d(\Bb)=\sum_{j=0}^{n-1} d_j$. Since $d(\Bb')<d_s(\Bb)$, it follows that there is an $i\in\N_{0,n-1}$ for which $d_i'<d_i$; i.e., $d_i-d_i'=\delta$ for a positive integer $\delta$.

By the first constraint of $\IPB$ and \eqref{eq_1}, it follows that the rows of $\Bb'$ are partitioned into $r$ groups $\{\K_\iota\}_{\iota=0}^{r-1}$ of size $\ell+1$; and $(s+1-r)$ groups $\{\K_\iota\}_{\iota=r}^{s}$ of size $\ell$. Without loss of generality, we assume that 
\begin{itemize}
  \item $\K_\iota=\{\iota+z\cdot(s+1):z\in\N_{0,\ell}\}$, for each $\iota\in\N_{0,r-1}$
  \item $\K_\iota=\{\iota+z\cdot(s+1):z\in\N_{0,\ell-1}\}$, for each $\iota\in\{r,r+1,\ldots,s\}$
\end{itemize}
i.e., $\{\K_\iota\}_{\iota=0}^s=\Cfr_1\bigsqcup\Cfr_2$; where $\Cfr_1\equiv\bigsqcup_{\iota=0}^{r-1}\K_\iota$ and $\Cfr_2\equiv\bigsqcup_{\iota=r}^s\K_\iota$ --- this assumption can met by a simple permutation on the rows of $\Bb'$; which does not affect $d_s(\Bb')$ nor any of the constraints on $\IPB$.

We first consider the case where $d_i'<d_i$; for some $i\in\Cfr_1$. Recall that Algorithm \ref{B_ones_unb_unif_C1} can be reduced to only include the \textbf{else if} statement, as when $\ell>s-r$ we have $\lambda=s$ and $\rt=\ell+r-s>0$. Thus, the two \textbf{if} loops are equivalent for $\ell>s-r$. It therefore suffices to only consider the \textbf{else if} statement of the algorithm. We know that $d_i=\big|\|\Bb_{(i)}\|_0-(s+1)\big|\in\big\{|\lambda-(s+1)|,|\lambda+1-(s+1)|\big\}$. When $\lambda=\|\Bb_{(i)}\|_0\leq s$; it follows that $\|\Bb_{(i)}'\|_0=\|\Bb_{(i)}\|-\delta$, and in order to meet \eqref{cond_gen_GC_B}; there is at least one $j\in[i]_{s+1}$ for which $\|\Bb_{(j)}\|_0=\lambda+1$ and $\|\Bb_{(j)}'\|_0\geq\|\Bb_{(j)}\|_0+1$. Therefore
\begin{equation}
  \big|\|\Bb_{(j)}'\|_0-\|\Bb_{(i)}'\|_0\big| = \|\Bb_{(j)}'\|_0-\|\Bb_{(i)}'\|_0 \geq \lambda+2-\big(\|\Bb_{(i)}\|_0-\delta\big) = \delta+2>1
\end{equation}
which violates the second constraint of $\IPB$. By a symmetric argument, one shows a similar contradiction for when $\|\Bb_{(i)}\|_0=\lambda+1\geq s+1$.

Next, we consider the case where $d_i'<d_i$ for $i\in\Cfr_2$. When $q=0$, we have $\|\Bb_{(j)}\|_0=s+t+1$ and $d_j=|s+t+1-(s+1)|=t$ for all $j\in[i]_{s+1}$, thus $d_s(\Bb)=\ell\cdot t$. Note that we cannot have $d_i'=d_i-\delta$ for $\delta\geq t+1$, as this would imply that
\begin{equation}
  \big|\|\Bb_{(i)}'\|-(s+1)\big|=\big|\|\Bb_{(i)}\|_0-(s+1)\big|-\delta = \big|s+t+1-(s+1)\big|-\delta \leq t-(t+1) = -1
\end{equation}
a contradiction. We therefore restrict this difference to $\delta\in\{1,2,\ldots,t\}$, for which it follows that $\|\Bb_{(i)}'\|_0=\|\Bb_{(i)}\|-\delta$. In order to meet \eqref{cond_gen_GC_B}, there is at least one $j\in[i]_{s+1}$ for which $\|\Bb_{(j)}'\|_0\geq\|\Bb_{(j)}\|_0+1$, thus
\begin{equation}
\label{contr_delta_C2}
    \big|\|\Bb_{(j)}'\|_0-\|\Bb_{(i)}'\|_0\big| = \|\Bb_{(j)}'\|_0-\|\Bb_{(i)}'\|_0 \geq \|\Bb_{(j)}\|_0+1 - \big(\|\Bb_{(i)}\|_0-\delta\big) = \delta+1>1
\end{equation}
which violates the second constraint of $\IPB$.

Lastly, we consider the case where $i\in\Cfr_2$, and $q>0$. In the case where $\|\Bb_{(i)}\|_0=s+t+1$, the argument is the same as above. In the case where $\|\Bb_{(i)}\|_0=s+t+2$ and $\delta=1$, by the third constraint of $\IPB$ it follows that there is at least one $j\in[i]_{s+1}$ for which $\|\Bb_{(j)}'\|_0>\|\Bb_{(j)}\|+1$, implying that $d(\Bb')\geq d(\Bb)$, contradicting the assumption that $d(\Bb')<d(\Bb)$. Hence, for the case where $\|\Bb_{(i)}\|_0=s+t+2$, the only differences $\delta$ we need to consider are $\delta\in\{2,3,\ldots,t+1\}$. This though, reduces to the same argument as above, which led to the contradiction in \eqref{contr_delta_C2}.

We therefore conclude that any $\Bb'\{0,1\}^{n\times k}$ for which $d_s(\Bb')<d_s(\Bb)$, violates at least one constraint of $\IPB$. Therefore, our construction of $\Bb$ through Algorithms \ref{B_ones_unb_unif_C1} and \ref{B_ones_unb_unif_C2} is a solution to $\IPB$, the binary version of $\OPR$.
\end{proof}

\begin{proof}{[Theorem \ref{cond_B_thm}]}
Assume condition 1) holds. By our construction of $\ab_\I$, we consider each congruence class separately. The superposition of the rows corresponding to a complete residue system $[i]_{s+1}$, is equal to the sum of these rows over $\R$. We denote this superposition for the $i^{th}$ congruence class by $\bar{\bb}_{[i]}$, i.e.,
\begin{equation}
\label{superp_vec_b_bar}
  \bar{\bb}_{[i]}\coloneqq\left(\sum_{\iota\in[i]_{s+1}}\bar{\Bb}_{(\iota)}\right),
\end{equation}
for which $\ab_{\I}^T\bar{\Bb}=\bar{\bb}_{[i]}$ if $\{\iota:\iota\in[i]_{s+1}\}\subsetneq\I$.
Since the vectors are binary, the superposition results in $\bold{1}_{1\times k}$ only when $1$ appears in each position in a single row of this congruence class. This is precisely condition 2); for $\bar{\Bb}$ satisfying 1).

Now, assume condition 2) holds. For binary rows $\{\bar{\Bb}_{(j)}\}_{j=1}^n$, condition 1) ensures that the cardinality of each of these vectors is equal to that of the corresponding row $\Bb_{(j)}$, i.e., the same number of partitions are allocated to the $j^{th}$ worker through both $\Bb$ and $\bar{\Bb}$. Therefore, for $\bar{\Bb}$ satisfying 1), we get
\begin{equation}
  \bar{\bb}_{[i]} = \left(\sum_{\iota\in[i]_{s+1}}\Bb_{(\iota)}\right) = k\ .
\end{equation}
for all $i\in\{0,\cdots,s\}$. Under the assumption that 2) is satisfied, we have $(\bar{\bb}_{[i]})_l\in\{0,1\}$ for all $l\in\N_k$, thus $\bar{\bb}_{[i]}=\bold{1}_{1\times k}$. Specifically, we saw that for $\bar{\Bb}\in\{0,1\}^{n\times k}$ satisfying condition 2), we have $\bar{\bb}_{[i]}\in\{0,1\}^{1\times k}$ for all $i\in\N_{0,s}$. When condition 1) is also satisfied, we then have $\bar{\bb}_{[i]}=\ab_\I^T\bar{\Bb}=\bold{1}_{1\times k}$ for $i$ such that $\{\iota:\iota\in[i]_{s+1}\}\subsetneq\I$. We conclude that if 1) and 2) are simultaneously satisfied, the first statement holds.

Condition 2) guarantees that $\|\bar{\Bb}^{(i)}\|_0\leq s+1$ for all $i$. Since we are applying a permutation on each set of rows corresponding to a complete residue system separately, we get that $\|\bar{\Bb}^{(i)}\|_0\geq\|\Bb^{(i)}\|_0$, and by our construction of $\Bb$, we are guaranteed that $\|\Bb^{(i)}\|_0=s+1$ for all $i$. By antisymmetry, it is clear that $\|\bar{\Bb}^{(i)}\|_0=s+1$ for all $i$. This completes the proof.
\end{proof}

\section{Numerical Experiments}
\label{num_expers_app}

\subsection{Coded vs. Uncoded}

In the following experiment, we justify the benefit of encoding computations in distributed platforms. We considered the fastest 250 AWS (Amazon Web Services) server completion times from \cite{BP23} to model the delays of our experiment. Similar experiments have been considered in other works, e.g., \cite{TLDK17,HASH17,LLPPR18}, though these consider artificially delayed stragglers, with significantly smaller $n$ and $s$. The experiment presented below was conducted multiple times, and is a representative example of these replicated experiments.

We considered $A\in\R^{L\times N}$ and $B\in\R^{N\times M}$ for $L=M=N=10^4$, which we partition across the respective dimension $N$ to accommodate our scheme $\cmmO$ with $n=250$ workers; to tolerate $s$ stragglers. Specifically, deployed $\cmmO$ with blocks of size $\tau=\frac{N}{k}=20$ for $k=500$, and in the coded setting, each worker was assigned $s+1$ different blocks, where $s$ varied for different experiments. Once the computation times were calculated, we added the delay times from the AWS server completion times \cite{BP23} mentioned above. We ordered the delay times in ascending order, and note that there was significant difference between the responses of workers 184 and 185 (3.325 seconds), and workers 238 and 239 (7.1463 seconds). The $1^{st}$ and $250^{th}$ fastest workers in \cite{BP23} responded after 5.965 and 20.841 seconds respectively. For $T_i$ the computation time of the $i^{th}$ worker from our computation according to $\cmmO$ and $W_i$ the response time according to the empirical distribution from \cite{BP23}, we emulated the overall waiting time by the central server for worker $i$ as: $T_i+(W_i-5.965)$.

In Table \ref{stragg_CMM1_table}, we report the respective response time of the slowest worker which was needed in order to recover the matrix product (i.e., waiting time of the $n-s$ fastest worker; which corresponds to the recovery threshold of $\cmmO$), and the slowest worker of the distributed computation when no coding was applied. It is worth mentioning that our approach could have a greater speed up when the decoding step of  Algorithm \ref{stream_dec} were to be used instead. In this experiment, we report the worst case scenario of our approach. Since these trials were carried out on the same personal computer, we expect that the times reported in the ``Uncoded'' row should be the same. They differ slightly, as the matrix product corresponding to each column, was for different random matrices $A$ and $B$. Our approach was beneficial in the case where $\cmmO$ was deployed with $s\in\{6,8,10,12,14,16\}$, which is a consequence of the delay times of the slower servers.

\begin{center}
\begin{table}[h]
\centering
\begin{tabular}{ |p{1cm}||p{.8cm}|p{.8cm}|p{.8cm}|p{.8cm}|p{.8cm}|p{.8cm}|p{.8cm}|p{.8cm}| }
\hline
\multicolumn{9}{|c|}{Emulated AWS Recovery Times for $\cmmO$} \\
\hline
\hline
$\quad s$ & $2$ & $4$ & $6$ & $8$ & $10$ & $12$ & $14$ & $16$ \\
\hline
$\cmmO$ & \textbf{21.0334} & \textbf{20.9271} & \textbf{20.7375} & \textbf{20.623} & \textbf{20.603} & \textbf{13.2823} & \textbf{13.0001} & \textbf{12.7439} \\
\hline
Uncoded & 21.093 & 21.0017 & 21.5725 & 21.0055 & 21.0527 & 21.3964 & 21.0164 & 21.092 \\
\hline
\end{tabular}
\vspace{3mm}
\caption{Emulated AWS response times, for $\cmmO$ and uncoded distributed matrix multiplication. We report the waiting times of the slowest responsive worker we need in order to perform $\cmmO$; i.e., the time of the $n-s$ fastest worker, and the slowest of the $250$ workers in the uncoded scenario. In bold, we indicate which of the two respective times was faster. The times reported are in seconds.}
\label{stragg_CMM1_table}
\end{table}
\end{center}
\vspace{-.8cm}

We also note that the decoding matrix $\Ab$ of the GCS from \cite{TLDK17}, for $(s=9,n=250)$ would be comprised of approximately $9.09\times 10^{15}$ total rows, while for $(s=24,n=250)$ we would have $1.83\times 10^{33}$ total rows, which are both infeasible to store and search through. We report these values for the given parameters of $s$, as the work of \cite{TLDK17} requires that $(s+1)\mid n$. Even in the simple case where $(s=5,n=30)$, a personal computer cannot store the resulting decoding matrix $\Ab$, which is of size $142506\times 30$.

\subsection{Numerical Error Experiment}
\label{num_stab_exp_sec}

A compelling motivation for using binary matrices for our GCS, is that they require low complexity for both encoding and decoding; and do not introduce numerical
nor rounding errors, compared to schemes with encoding matrices defined over the real or complex numbers. Another motivating factor, is the fact that the decoding step of \cite{TLDK17} constructs and stores a matrix $\Ab$ of size ${{n}\choose{s}}\times n$; which is comprised of the decoding vectors corresponding to each of the possible index sets $\I$. Searching through $\Ab$ to find the corresponding $\ab_{\I}$ is also prohibitive; as it has $\Theta(n^s)$ rows to look through

In this experiment, we compare the residual error $\|g-\gt\|_2$ for $g$ the gradient computed by a single server without any encoding or decoding taking place, and $\gt$ the gradient computed by the two competing GC schemes. We consider the case where $n=k=18$, $s=5$, and the gradient and partial gradients are of dimension $p=100$; with varying norm $\|g\|_2$. For these parameters, we have $\Ab\in\R^{8568\times 18}$.

Our GCS introduces no error through the encoding and decoding steps, i.e., $\|g-\gt\|_2=0$. This is due to the fact that no multiplication or division takes place by scalars, and that the decoding step turns out to be an addition of the exact partial gradients\footnote{By \textit{exact}, we mean $g_i$ as computed by a single server.}. For the FRC scheme, we considered a binary encoding matrix $\Bb$ whose construction is deterministic, which introduces no error. The construction of matrix $\Ab$, solves ${{n}\choose{s}}$ linear systems of the form $\Bb_{\I}^T\cdot y_{\I}=\bold{1}_{f\times k}$; for $\Bb_{\I}\in\{0,1\}^{f\times k}$ the submatrix of $\Bb$ corresponding to $\I$. This is done using matlab's backslash operation, which computes the decoding vector $\ab_{\I}\gets y_{\I}=(\Bb_{\I}^T)^{\dagger}\cdot\bold{1}_{1\times k}$ for each $\I$.

For this simple case where $n=k=18$ and $s=5$, out of the $8568$ index sets $\I$, the condition number $\kappa_2(\Bb_{\I})\coloneqq\frac{\sigma_{\max}(\Bb_{\I})}{\sigma_{\min}(\Bb_{\I})}$ of 596 submatrices $\Bb_{\I}$ is greater than matlab's finite floating-point number of $1.7977\times10^{308}$. The remaining 7972 $\Bb_{\I}$ submatrices had a condition number ranging between $1.7121\times10^{50}$ and $2.4610\times10^{276}$. Each of the systems $\Bb_{\I}^T\cdot y_{\I}=\bold{1}_{s\times k}$ are therefore ill-conditioned. On the other hand, the construction of our decoding vectors is an addition of standard basis vectors \eqref{vecs_ai}, which does not require solving a linear system.

In Figure \ref{FRC_error_plot} we show how the error introduced by the FRC scheme relates to the norm of the gradient. In machine learning applications where $p$ is large and many samples are considered; it is expected to have a large $\|g\|_2$, which results in larger error when the FRC scheme is deployed. Our scheme introduced no error, as shown in Figure \ref{binary_error_plot}.

\begin{figure}[h]
\hspace*{-.8cm}
 \centering
    \includegraphics[scale=.11]{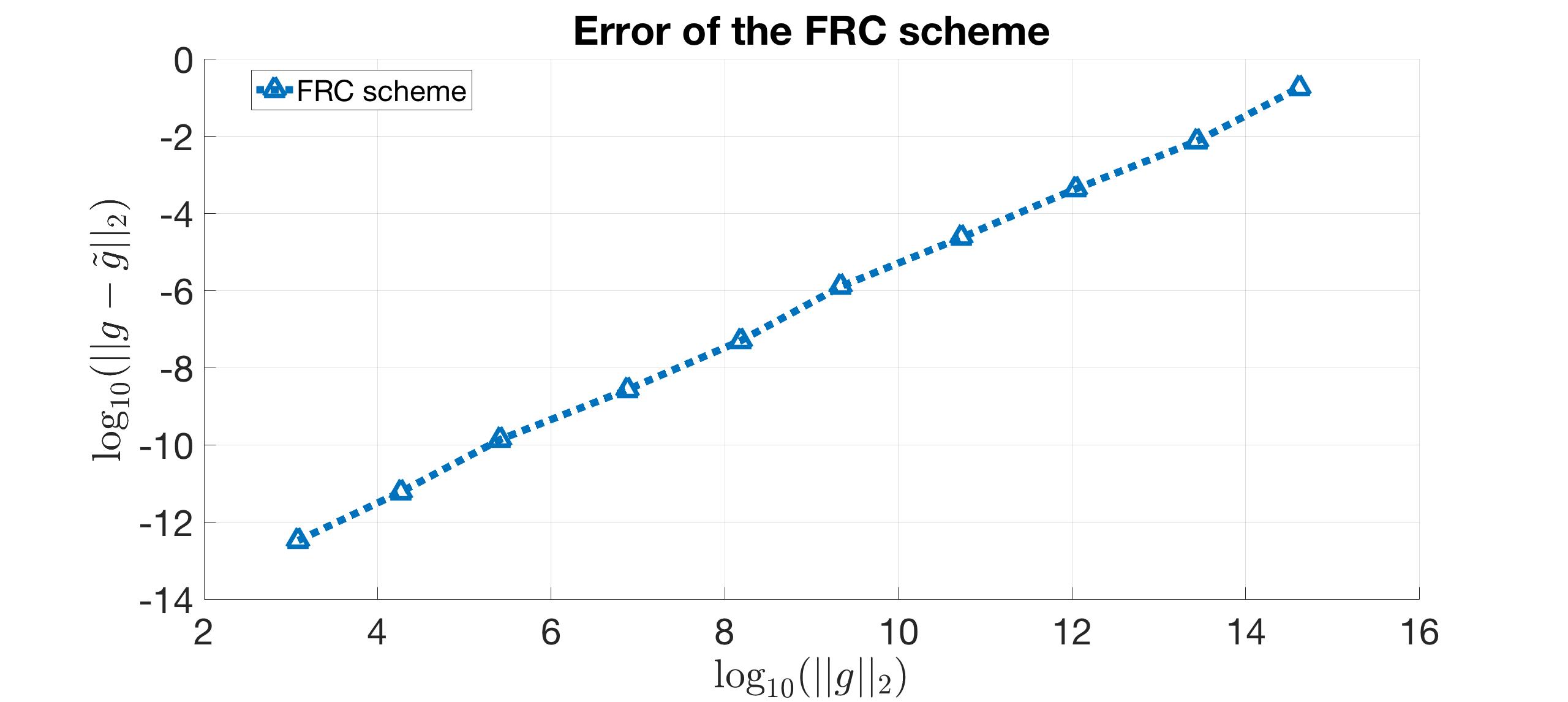}
  \caption{Propagation of error introduced by the FRC scheme, in log-scale.}
  \label{FRC_error_plot}
\end{figure}

\begin{figure}[h]
\hspace*{-.8cm}
 \centering
    \includegraphics[scale=.11]{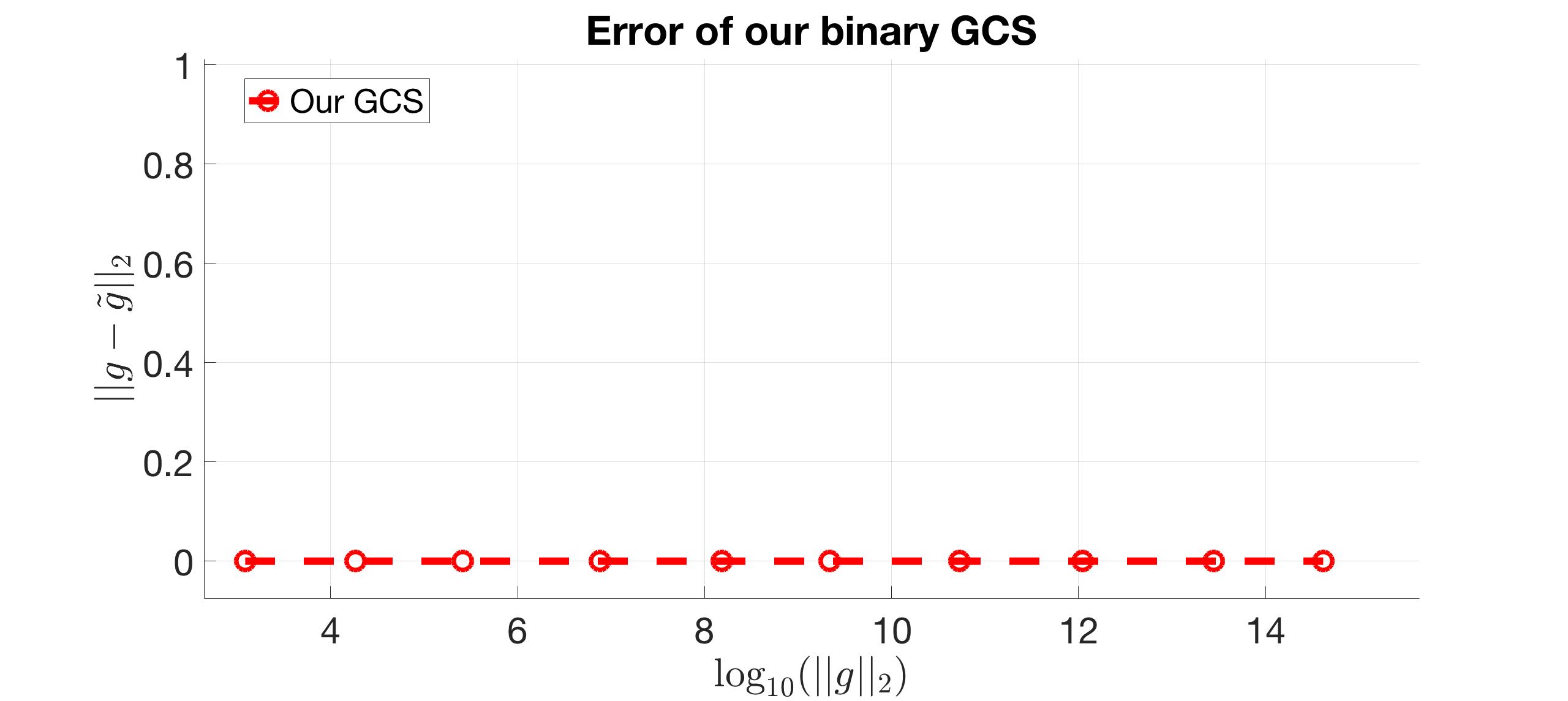}
  \caption{Error plot of our GCS.}
  \label{binary_error_plot}
\end{figure}

\section{Application of CMM to Distributed Gradient Descent for Frobenius-norm Minimization}
\label{appl_Fr_norm}

In this appendix we first review gradient descent, and then focus on gradient descent for Frobenius-norm minimization problems, as defined in \eqref{Th_star_pr}; for the objective function $L_F$ defined in \eqref{opt_pr_LF}. We briefly describe the following motivating problems which try to solve \eqref{Th_star_pr} or similar optimization problems: \textit{nonnegative matrix factorization} (NMF), $k$-$\SVD$, low rank matrix approximation, sparse coding and  the best $k$-rank approximation; which relates to principal component analysis. This is not an exhaustive list of where the objective function $L_F$ has been utilized, since many more applications do exist.

Recall that in gradient descent, we consider a minimization problem with a convex differentiable objective function $L\colon\Cc\to\R$ over an open constrained set $\Cc\subseteq \R^p$. Then, given an initial $\theta^{[0]}\in\Cc$; the following update is performed at iteration $t+1$:
\begin{equation}
  \theta^{[t+1]}\gets\theta^{[t]}-\xi_t\cdot\nabla_{\theta}L(\D;\theta^{[t]}), \quad \text{ for } t=0,1,2,\ldots
\end{equation}
until a specified termination criterion is met. The parameter $\xi_t$ is the step-size, which may be adaptive or fixed. Note that GC is only concerned with computing the gradient at each step and hence, selecting an appropriate step-size was not discussed in this paper.

In the literature regarding gradient descent for coded computing thus far, only the case where the gradient of the objective function \eqref{th_star_pr} is a vector has been considered or discussed. In order to tie together Sections \ref{str_problem_GC} and \ref{BGC_sec} with Section \ref{matr_mult_sec}, we discuss the case where the gradient is a matrix, e.g.,
\begin{equation}
\label{opt_pr_LF}
  L_F(\Xb,\Yb;\Theta) \coloneqq \|\Xb\Theta-\Yb\|_F^2 = \sum_{i=1}^m\overbrace{\|\Xb\Theta^{(i)}-\Yb^{(i)}\|_2^2}^{L_{ols}\left(\Xb,\Yb^{(i)};\Theta^{(i)}\right)},
\end{equation}
for $\Xb\in\R^{N\times p}$, $\Theta\in\R^{p\times m}$ and $\Yb\in\R^{N\times m}$. The gradient is
\begin{equation}
  \nabla_{\Theta} L_F(\Xb,\Yb;\Theta) = 2\Xb^T(\Xb\Theta-\Yb),
\end{equation}
which is computed in order to approximate the solution to
\begin{equation}
\label{Th_star_pr}
  \Theta^{\star} = \argmin_{\Theta\in\R^{p\times m}}\big\{L_F(\Xb,\Yb;\Theta)\big\}
\end{equation}
via gradient descent. Similar to the ordinary least squares objective function $L_{ols}$, \eqref{Th_star_pr} has the closed-form solution:
\begin{equation}
  \Theta^{\star} = \Xb^{\dagger}\Yb = (\Xb^T\Xb)^{-1}\Xb^T\Yb,
\end{equation}
which is intractable for large $N$. In practice, it is often preferred to approximate $\Theta^{\star}$.

A motivating application is if we have measurements $\{\Yb^{(i)}\}_{i=1}^m$ from $m$ different sensors in different locations or from different sources, for the same corresponding $\Xb$, and we want to interpolate the corresponding optimal models $\{\theta_i\}_{i=1}^m$ for each sensor or source.

More generally, the goal of most nonlinear regression problems is to solve the problem
\begin{equation}
  \min_{\kappa\in\mathcal{H}}\left\{\sum_{i=1}^N\big(\kappa(\xb_i)-y_i\big)^2\right\},
\end{equation}
where $\kappa$ comes from a hypothesis class $\mathcal{H}$ that fits the training data, for which one can use the kernel trick to solve efficiently. What we present can be applied also to regression problems of this type, as well as kernel regression problems \cite{SB14}.

Throughout the gradient descent process, the second summand $2\Xb^T\Yb$ is constant. Hence, at every iteration we only need to compute the matrix product $\Xbh\Theta$, where $\Xbh=2\Xb^T\Xb$ is also a constant matrix. Then, depending on which matrix multiplication scheme we decide to use, the workers will receive the entire matrix $\Xbh$ or a submatrix of it at the beginning of the distributed computation process, and a submatrix of $\Theta$'s update
\begin{equation}
  \Theta^{[t+1]}\gets\Theta^{[t]}-\xi_t\cdot\nabla_{\Theta}L_F(\Xb,\Yb;\Theta^{[t]})
\end{equation}
at each iteration. In an iterative process, it is preferred to reduce the total communication cost as much as possible. Hence, we prefer to communicate only part of $\Theta$ when possible.

Solving for the loss function $L_\Theta$ may also be viewed as solving multiple linear regression problems simultaneously. This is due to its decomposition into a summation of $m$ separate least squares objective functions, with the same data matrix $\Xb$. For $\Theta^{(i)}=\theta_i$, it follows that
\begin{equation}
  \Theta^{\star} = {\begin{pmatrix} | & | & & | \\ \theta_1^{\star} & \theta_2^{\star} & \hdots & \theta_m^{\star} \\ | & | & & | \end{pmatrix}} \in \R^{p\times m},
\end{equation}
for $\theta_i^{\star}$ being the minimal solution to $L_{ols}(\Xb,\Yb^{(i)};\Theta^{(i)})$, for each $i\in\N_m$. To guarantee convergence for all $\theta_i$, we can fix $\xi_t=2/\sigma_{\text{max}}(\Xb)^2$ for all iterations.

We point out that the above problem could indeed be solved by using regular GC, as we have
\begin{equation}
  \|\Ab\|_F^2 = \left\|\left[\big(\Ab^{(1)}\big)^T \ \cdots \ \big(\Ab^{(m)}\big)^T\right]^T\right\|_2^2\ ,
\end{equation}
for any real-valued matrix $\Ab$ comprised of $m$ columns. We also note that the least squares regression problem in the presence of stragglers, was studied in \cite{LKYSA18}. 

\subsection{Nonnegative Matrix Factorization}

The NMF problem deals with decomposing a matrix $A\in\R_{\geq0}^{L\times M}$ with nonnegative entries into two matrices $U\in\R_{\geq0}^{L\times N}$ and $V\in\R_{\geq0}^{N\times M}$, by attempting to solve
\begin{equation}
  \min_{\substack{U\in\R_{\geq0}^{L\times N}\\ V\in\R_{\geq0}^{N\times M}}}\big\{\|A-UV\|_F^2\big\}
\end{equation}
for $U,V$ with the appropriate dimensions. In \cite{LS01} a multiplicative update algorithm is proposed:
\begin{equation}
  V\gets V\cdot\frac{U^TA}{U^TUV} \qquad \text{ and } \qquad U\gets U\cdot\frac{AV^T}{UVV^T},
\end{equation}
where the division is done element-wise. Multiple multiplications are required for these updates, and the matrices can be quite large, e.g., when dealing with recommender systems. Multiple distributive multiplications are required at each iteration, which makes this process a lot more cumbersome. Therefore, speeding up this process is even more crucial. Further details on this algorithm and how to incorporate gradient methods to solve NMF can be found in \cite{Chi07a,Chi07b,EMWK18}.

\subsection{Low-Rank Approximation}

Consider the problem of finding a low-rank approximation of a matrix $A\in\R^{L\times M}$. That is, for an approximation of rank $k$ or less we want to find $B=U_BV_B$ for $U_B\in\R^{L\times k}$ and $V_B\in\R^{k\times M}$, which can be done by solving the problem
\begin{equation}
\label{low_rank_pr}
  \min_{\substack{B\in\R^{L\times M}\\ \rank(B)\leq k}}\big\{\|A-B\|_F^2\big\} = \min_{\substack{U_B\in\R^{L\times k}\\ V_B\in\R^{k\times M}}}\big\{\|A-U_BV_B\|_F^2\big\},
\end{equation}
where it is easier to work with $B=U_BV_B$ in the case where both $L$ and $M$ are large, e.g., in terms of storage and computations. 

The objective function of \eqref{low_rank_pr} is bi-convex. A common alternating minimization approach fixes one of the matrices and optimizes the other, and then alternates.

It is well-known that the best $k$-rank approximation for many norms, including the Euclidean and Frobenius norms, can be computed through the truncated singular value decomposition ($\SVD$). By the Eckart–Young theorem \cite{EY36}
\begin{equation}
  A_k = U \Sigma_k V^T = \sum_{i=1}^k\sigma_iU^{(i)}(V^{(i)})^T
\end{equation}
solves \eqref{low_rank_pr}. The $\SVD$ takes $O(LM\cdot\min\{L,M\})$ time itself to compute, which is cumbersome. To avoid computing the $\SVD$, we can resort to approximating $A_k$ by solving
\begin{equation}
  \min_{\substack{U\in\R^{L\times k}\\ U^TU=\Ib_k}}\big\{\|A-UU^TA\|_F^2\big\} = \min_{\substack{U\in\R^{L\times k}\\ U^TU=\Ib_k}}\big\{-\tr(U^TAA^TU)\big\}
\end{equation}
through gradient descent, where the gradient with respect to $U$ is $-AA^TU$. Hence, at each iteration we distributively compute $AA^TU^{[t]}$ \cite{SB14,KHFM19}. For $\hat{U}$ our final solution to the above minimization problem, our $k$-rank approximation of $A$ will be $\hat{U}\hat{U}^TA$, which is an approximation of $A_k$. For $U^\star$ being the exact solution, we have $A_k=U^\star({U^\star})^TA$.

Other problems in which CMM could be utilized in distributed gradient methods are the weighted low-rank matrix approximation \cite{BWZ19} and the $k$-$\SVD$ algorithm \cite{AEB06}. These involve similar objective functions, whose gradients have a matrix form which require at least one matrix-matrix multiplication. Thus, the process would be accelerated if these were to be computed distributively.


\balance
\bibliographystyle{IEEEtran}
\bibliography{refs_all}

\end{document}